\documentclass[11pt]{article}

\usepackage[american]{babel}

\usepackage{dsfont}
\usepackage{amsthm}
\usepackage{amsmath}
\usepackage{amssymb} 
\usepackage{xspace}
\usepackage{color}
\usepackage{enumerate}

\usepackage{thmtools, thm-restate}
\declaretheorem[numberwithin=section]{theorem}

\declaretheorem[sibling=theorem]{lemma}
\declaretheorem[sibling=theorem]{corollary}
\declaretheorem[sibling=theorem]{claim}
\declaretheorem[sibling=theorem]{remark}
\declaretheorem[sibling=theorem]{example}

\usepackage{multirow}

\usepackage[latin1]{inputenc}

\newcommand{\R}{\mathbb{R}}
\newcommand{\Z}{\mathbb{Z}}
\newcommand{\N}{\mathbb{N}}

\newcommand{\Ex}{\mathbb{E}}

\newcommand{\x}[1]{\ensuremath{\mathsf{x}_{#1}}}
\newcommand{\w}[1]{\ensuremath{\mathsf{w}_{#1}}}
\newcommand{\W}{\ensuremath{\mathsf{W}}}
\newcommand{\wmin}{\ensuremath{\mathsf{w}_{\min}}}
\newcommand{\wmax}{\ensuremath{\mathsf{w}_{\max}}}

\newcommand{\barw}{\ensuremath{\overline w}}
\newcommand{\barV}{\ensuremath{\bar V}}
\newcommand{\barG}{\ensuremath{\bar G}}

\def\dist{\operatorname{dist}}

\newcommand{\GIRG}{\mathcal{G}}
\newcommand{\Space}{\mathcal{X}}

\renewcommand{\epsilon}{\ensuremath{\varepsilon}}
\newcommand{\eps}{\ensuremath{\varepsilon}}

\newcommand{\temporary}[1]{}

\makeatletter
\newcommand{\pushright}[1]{\ifmeasuring@#1\else\omit\hfill$\displaystyle#1$\fi\ignorespaces}
\newcommand{\pushleft}[1]{\ifmeasuring@#1\else\omit$\displaystyle#1$\hfill\fi\ignorespaces}
\makeatother

\begin{document}

\title{Average Distance in a General Class of Scale-Free Networks \\ with Underlying Geometry}
\author{Karl Bringmann\thanks{Max-Planck-Institute for Informatics, Saarbr\"ucken, Germany, \texttt{kbringma@mpi-inf.mpg.de}} \and Ralph Keusch\thanks{Institute of Theoretical Computer Science, ETH Zurich, Switzerland, \texttt{rkeusch@inf.ethz.ch}} \and Johannes Lengler\thanks{Institute of Theoretical Computer Science, ETH Zurich, Switzerland, \texttt{lenglerj@inf.ethz.ch}}}
\date{}
\maketitle

\medskip

\begin{abstract}
~In Chung-Lu random graphs, a classic model for real-world networks, each vertex is equipped with a weight drawn from a power-law distribution, and two vertices form an edge independently with probability proportional to the product of their weights. Chung-Lu graphs have average distance $O(\log\log n)$ and thus reproduce the small-world phenomenon, a key property of real-world networks. 
Modern, more realistic variants of this model also equip each vertex with a random position in a specific underlying geometry. The edge probability of two vertices then depends, say, inversely polynomial on their distance. 

In this paper we study a generic augmented version of Chung-Lu random graphs. We analyze a model where the edge probability of two vertices can depend arbitrarily on their positions, as long as the marginal probability of forming an edge (for two vertices with fixed weights, one fixed position, and one random position) is as in Chung-Lu random graphs. The resulting class contains Chung-Lu random graphs, hyperbolic random graphs, and geometric inhomogeneous random graphs as special cases. 

Our main result is that every random graph model in this general class has the same average distance as Chung-Lu random graphs, up to a factor $1+o(1)$. This shows in particular that specific choices, such as the underlying geometry being Euclidean or the dependence on the distance being inversely polynomial, do not significantly influence the average distance.
The proof also shows that every random graph model in our class yields a giant component and polylogarithmic diameter with high probability. 
\end{abstract}



\bigskip

\section{Introduction}
Large real-world networks, like social networks or the internet infrastructure, are almost always \emph{scale-free}, i.e., their degree distribution follows a power law with parameter $2<\beta<3$. 
Such networks have been studied in detail since the 60s. One of the key findings is the \emph{small-world phenomenon}, which is the observation that two nodes in a network typically have very small graph-theoretic distance. In the 90s, this phenomenon was explained by theoretical models of random graphs. Since then, random graph models have been the basis for the statistical study of real-world networks, as they provide a macroscopic perspective and reproduce structural properties observed in real data. In this line of research, one studies the diameter of a graph, i.e., the largest distance between any pair of vertices in the largest component, and its average distance, i.e., the expected distance between two random nodes of the largest component. A random graph model is said to be \emph{small world} if its diameter is bounded by $(\log n)^{O(1)}$ or even $O(\log n)$, and \emph{ultra-small world} if its average distance is only $O(\log \log n)$. 

Chung-Lu random graphs are a prominent model of scale-free networks \cite{chung2002avg,Chung02}. In this model, every vertex $v$ is equipped with a weight $\w v$, and two vertices $u,v$ are connected independently with probability $\min\{1, \w u \w v / \W\}$, where $\W$ is the sum over all weights $\w v$.  The weights are typically assumed to follow a power-law distribution with power-law exponent $\beta>2$. Chung-Lu random graphs have the ultra-small world property, since in the range $2<\beta<3$ the average distance is $(2 \pm o(1)) \frac{\log \log(n)}{|\log(\beta -2)|}$ \cite{chung2002avg,Chung02}.

However, Chung-Lu random graphs fail to capture other important features of real-world networks, such as high clustering or navigability. This is why dozens of papers propose more realistic models which also possess some local structure, many of which combine Chung-Lu random graphs (or other classic models such as preferential attachment~\cite{BarabasiA99}) with an underlying geometry, see, e.g., hyperbolic random graphs~\cite{BogunaPK10,PapadopoulosKBV10,gugelmann2012random}, geometric inhomogeneous random graphs~\cite{bringmann2015euclideanGIRG,koch2016bootstrap,bringmann2017routing}, and many others~\cite{spatialpreferred,BollobasSR07,bonato2010geometric,bradonjic2008structure,
deijfen2013scale,jacob2013spatial,SerranoKB08}. 
In these models, each vertex is additionally equipped with a random position in some underlying geometric space, and the edge probability of two vertices depends on their weights as well as the geometric distance of their positions. Typical choices for the geometric space are the unit square, circle, or torus, and for the dependence on the distance are inverse polynomial, exponential, or threshold functions. 
Such models can naturally yield a large clustering coefficient, since there are many edges among geometrically close vertices.
For some of these models the average distance has been studied and shown to be the same as in Chung-Lu graphs, up to a factor $1+o(1)$, see, e.g.,~\cite{abdullah2015typical,BollobasSR07,deijfen2013scale}.

For these results, it is unclear how much they depend on the particular choice of the underlying geometry. In particular, it is not known whether any of the important properties of Chung-Lu random graphs transfer to versions with a \emph{non-metric} underlying space. Such spaces are well-motivated in the context of social networks, where two persons are likely to know each other if they share a feature (e.g., they are in the same sports club) regardless of their differences in other features (e.g., their profession), which gives rise to a non-metric distance (see Section~\ref{sec:distancemodel}). 

\paragraph{Our Contribution:}
As main result of this paper we prove that all geometric variants of Chung-Lu random graphs have the same average distance $(2 \pm o(1)) \frac{\log \log(n)}{|\log(\beta -2)|}$ in the regime $2<\beta<3$, showing \emph{universality} of the ultra-small world property. 

We do this by analyzing a generic augmented and very general version of Chung-Lu random graphs. Here, each vertex is equipped with a power-law weight $\w v$ and an independently random position $\x v$ in some ground space~$\Space$. Two vertices $u,v$ form an edge independently with probability $p_{uv}$ that only depends on the positions $\x u, \x v$ (and $u,v$ and the weight sequence). The dependence on $\x u, \x v$ may be arbitrary, as long as the edge probability has the same marginal probabilities as in Chung-Lu random graphs. Specifically, for fixed $\x u$ and random $\x v$ we require that the marginal edge probability $\Ex_{\x v}[p_{uv} | \x u]$ is within constant factors of the Chung-Lu edge probability $\min\{1, \w u \w v / W\}$. This is a natural property for any augmented version of Chung-Lu random graphs. Note that our model is stripped of any geometric specifics. In fact, the ground space is not even required to be metric. We retain only the most important features, namely power-law weights and the right marginal edge probabilities. Hence, the main result also demonstrates that there exist random graph models with non-metric underlying geometry, still satisfying the ultra-small world property.

Beyond the average distance, we establish that this general model is scale-free and has a giant component and polylogarithmic diameter. Thus all instantiations of augmented Chung-Lu random graphs share some basic properties that are considered important for models of real-world networks. 

It it quite surprising that the average distance can be computed so precisely in this generality. For example the clustering coefficient varies drastically between different instantiations of the model, as it encompasses the classic Chung-Lu random graphs that have clustering coefficient $n^{-\Omega(1)}$, as well as geometric variants that have constant clustering coefficient~\cite{bringmann2015euclideanGIRG}. Therefore, our results hold on graphs with very different local structure. Note that by the scale-free-property, all variants of the model contain $\Theta(n)$ edges. If an instance has high clustering, many edges are \emph{local} edges inside well-connected subgraphs, and therefore futile for finding short paths between far vertices. Still, our main result implies that in such graphs the average distance is asymptotically the same as in Chung-Lu random graphs, where we have no clustering and every edge is potentially helpful when searching for short paths. We also remark that the statements fail to hold for $\beta > 3$, and that the graphs can look rather diverse depending on the model. For example, some instantiations in this regime do not even have a giant component, but the largest component is of polynomial size $n^{1-\Omega(1)}$~\cite{bode2014geometrisation}. On the other hand, it is also not hard to construct models for $\beta >3$ which do have a giant component, but still have polynomially large average distance, see Remark~\ref{rem:largebeta}. This variety for $\beta >3$ makes it even more surprising that in the regime $2 < \beta < 3$ the average distance can be determined precisely for all instances at once.

A common property of all models in our general class is that for a set $S$ of vertices whose weights sum to $W_S$ (often called \emph{volume} in the literature), the expected number of half-edges going out from $S$ is $\Theta(W_S)$. For the classic Chung-Lu random graphs without geometry, the targets of these half-edges are independent of each other. Thus the quantity $W_S$ is essentially sufficient to determine the size and the volume of the neighborhood $\Gamma(S)$ of $S$ and the analyses of Chung-Lu random graphs are based on this property. However, for non-trivial geometries the size of the neighborhood crucially depends on the geometric position of the vertices in $S$. For example, if the clustering coefficient is constant, then even if $S$ consists of only two adjacent vertices there is already a non-negligible probability that they share some neighbors. Thus the proofs for classic Chung-Lu random graphs do not carry over to the general setting. On the other hand, existing proofs for geometric scale-free networks~\cite{abdullah2015typical, deijfen2013scale} rely rather heavily on the specifics of the underlying geometry.

In the general setting, we can therefore only borrow one step from previous proofs, namely the ``greedy path'' argument (Lemma~\ref{lem:path-to-core}). We use this idea to prove that for all vertices of at least poly-logarithmic weight there exists an ultra-short path to the ``heavy core'', which is well-connected and contains the vertices of highest weight. From a technical point of view, the most important contribution of this paper is the ``bulk lemma'' (Lemma~\ref{lem:path-to-core3}). It contains a delicate and subtle analysis of the neighborhood of a vertex restricted to small-weight vertices. The lemma studies the probability that the $k$-neighborhood of a random vertex $v$ of low-weight contains a node $v'$ that is connected to a high-weight vertex, from which we can then apply the ``greedy path'' argument. We emphasize that both the size and the shape of the $k$-neighborhood of such a vertex crucially depend on the underlying geometry. Therefore, we are forced to use the geometry \emph{implicitly}, in order to make the argument general enough for being valid universally in our general class of random graphs. 
Finally, we obtain the bound on the average distance by applying the bulk lemma repeatedly for different values of $k$ and carefully summing up the resulting terms.

\paragraph{Organization of the Paper:}
In Section~\ref{sec:modelresults} we present the details and a precise definition of the model, and we formally state the results. In Section~\ref{sec:notation} we introduce notation and prove a concentration inequality which will be used later in the proofs. After some basic and preliminary results (Section~\ref{sec:basic}), we prove the connectivity properties and the main result in Section~\ref{sec:diameter}, and determine the degree distribution of our model in Section~\ref{sec:degreesequence}. We discuss several special cases of the model in Section~\ref{sec:distancemodel}, and make some concluding remarks in Section~\ref{sec:conclusion}.

\section{Model and Results}
\label{sec:modelresults}
\subsection{Definition of the Model}\label{sec:model}

In this paper we study properties of a very general random graph model, where both the set of vertices $V$ and the set of edges $E$ are random. Each vertex $v$ comes with a weight $\w v$, which will essentially be the expected degree of $v$, and with a random position $\x v$ in a geometric space $\Space$. We now give the full definition, first for the weight sequence and then for the resulting random graph.

\paragraph{Power law weights:}
For $n \in \N$ let $\w{} = (\w 1, \ldots, \w n)$ be a non-increasing sequence of positive weights. We call $\W := \sum_{v=1}^n \w v$ the \emph{total weight}. 
Throughout this paper we will assume that the weights follow a \emph{power law}: the fraction of vertices with weight at least $w$ is $\approx w^{1-\beta}$ for some $\beta > 2$ (the \emph{power-law exponent} of $\w{}$). More precisely, we assume that for some $\barw = \barw(n)$ with $n^{\omega(1/\log\log n)}\leq\barw \leq n^{(1-\Omega(1))/(\beta-1)}$, the sequence $\w{}$ satisfies the following conditions:
\begin{enumerate}[(PL1)]
\item the minimum weight is constant, i.e., $$\wmin := \min\{\w{v} \mid 1 \le v \le n\} = \Omega(1);$$ 
\item for all $\eta >0$ there exist constants $c_1,c_2>0$ such that
$$
c_1\frac{n}{w^{\beta-1+\eta}} \leq \#\{1 \le v \le n \mid \w{v} \geq w\} \leq c_2\frac{n}{w^{\beta-1-\eta}},
$$
where the first inequality holds for all $\wmin \leq w \leq \barw$ and the second for all $w \geq \wmin$. 
\end{enumerate}
We remark that these are standard assumptions for power-law graphs with average degree
$\Theta(1)$. Note that since $\barw \le
n^{(1-\Omega(1))/(\beta-1)}$, there are $n^{\Omega(1)}$ vertices with weight at
least $\barw$. On the other hand, no vertex has weight larger than $(c_2
n)^{1/(\beta-1-\eta)}$. 

\paragraph{Random graph model:}
Let $\Space$ be a non-empty set, and assume we have a measure $\mu$
on $\Space$ that allows to sample elements from $\Space$. We call $\Space$ the \emph{ground space} of the model and the elements in $\Space$ \emph{positions}.
The random graph $\GIRG(n,\Space,\w{},p)$ has vertex set
$V=[n] = \{1,\ldots,n\}$. For any vertex $v$ we independently draw a position $\x{v} \in \Space$ according to measure $\mu$. Conditional on $\x{1}, \ldots, \x{n}$, we connect any two
vertices $u \ne v$ independently with probability
$$p_{uv} := p_{uv}(\x{u},\x{v}) :=
p_{uv}(\x{u},\x{v};n,\Space,\w{}),$$ 
where $p$ is a (symmetric in $u,v$ and measurable) function mapping to $[0,1]$ that satisfies the following condition:
\begin{enumerate}[(EP1)]
  \item for any $u,v$, if we fix position $\x u \in \Space$ and draw position $\x v$ from $\Space$ according to $\mu$, then the marginal edge probability is
  $$
   \Ex_{\x v}[p_{uv}(\x u, \x v) \mid \x u] = \Theta\Big(\min\Big\{1,\frac{\w u \w v}{\W}\Big\}\Big).
  $$ 
\end{enumerate}

\noindent For most results we also need an additional condition, to ensure the existence of a unique giant component:
\begin{enumerate}[(EP2)]
  \item for all $\eta > 0$, any $u,v$ with $\w{u},\w{v} \geq \barw$, and any fixed positions $\x{u},\x{v} \in \Space$ we have
  $$
   p_{uv}(\x u, \x v) \geq \Big(\frac{n}{\barw^{\beta-1+\eta}}\Big) ^{-1+\omega(1/\log\log n)}.
  $$
\end{enumerate}

\paragraph{Discussion of the model:}
Let us first argue why condition (EP2) is necessary to obtain a unique giant component. Suppose we
have an instantiation of our model $G$ on a space $\Space$. We will see in this paper that
with high probability $G$ has a giant component that contains all high-degree
vertices. Now make a copy $\Space'$ of $\Space$, and consider a graph where all
vertices draw geometric positions from $\Space \cup \Space'$. Vertices in
$\Space$ are never connected to vertices in $\Space'$, but within $\Space$ and
$\Space'$ we use the same connection probabilities as for $G$. Then the
resulting graph will satisfy all properties of our model except for (EP2),
but it will have two giant components, one in $\Space$ and one in $\Space'$. As
we will see, (EP2) ensures that the high-weight vertices form a single dense
network, so that the graph indeed has a unique giant component. However, for our results on the degree sequence (EP2) is not necessary.

Since the right hand side of (EP1) is the edge probability of Chung-Lu graphs, this is a natural condition for any augmented version of Chung-Lu graphs. In particular, (EP1) ensures that the expected degree of a vertex $v$ with weight $\w v$ is indeed $\Theta(\w v)$.
For similar reasons as discussed for (EP2), we cannot further relax (EP1) to a condition on the marginal probability over random positions $\x u$ and $\x v$, i.e, a condition like $\Ex_{\x u, \x v}[p_{uv}(\x u, \x v)] = \Theta\left(\min\left\{1,\frac{\w u \w v}{\W}\right\}\right)$. Indeed, consider the same setup as above, with $G$, $\Space$, and copy $\Space'$. For two vertices of weight at most $\bar w$, connect them only if they are in the same copy of $\Space$. For two vertices of weight larger than $\bar w$, always treat them as if they would come from the same copy (then condition (EP2) is satisfied). For a vertex $u$ of weight at most $\bar w$ and $v$ of weight larger than $\bar w$, connect them only if $u$ is in $\Space'$. Then the high-weight vertices form a unique component, but it is only connected to vertices in $\Space'$, while the low-weight vertices in $\Space$ may form a second giant component. Thus, in (EP1) it is necessary to allow any fixed $\x u$.

\paragraph{Sampling the weights:} In the definition we assume that the weight
sequence $\w{}$ is fixed. However, if we sample the weights according to an appropriate distribution, then the sampled weights will follow a power law with probability $1 - n^{-\Omega(1)}$, so that a model with sampled weights is almost surely included in our model. For the precise statement, see
Lemma~\ref{lem:sampleweights}.

\paragraph{Examples:}
We regain the Chung-Lu model as a special case by setting $\Space = \{x\}$ (the trivial
ground space) and $p_{uv} = \min\left\{1,\frac{\w u \w v}{\W}\right\}$, since then (EP1) is  trivially satisfied and (EP2) is satisfied for $2 < \beta < 3$.

We discuss more examples in Sections~\ref{sec:distancemodel}. 
In particular, the model includes \emph{geometric inhomogeneous random graphs} (GIRGs) that were introduced in~\cite{bringmann2015euclideanGIRG}. Consider the $d$-dimensional ground space $\Space=[0,1]^d$ with the standard (Lebesgue) measure, where $d \ge 1$ is a
(constant) parameter of the model. Let $\alpha \neq 1$ be a second
parameter that determines how strongly the geometry influences edge probabilities. Finally, let $\|.\|$ be the Euclidean distance on $[0,1]^d$, where we identify $0$ and $1$ in each coordinate (i.e., we take the distance on the torus). We show in Theorem~\ref{thm:distancemodel} that every edge probability function $p$ satisfying
\begin{equation}\label{eq:Euclideanmodel}
  p_{uv} = \Theta\Big( \min\Big\{1, (\|\x u-\x v\|)^{-d\alpha} \cdot \Big( \frac{\w u \w v}{\W}\Big)^{\max\{\alpha,1\}} \Big\} \Big)
\end{equation}
follows~(EP1) and (EP2), so it is a special case of our model. As was shown in~\cite{bringmann2015euclideanGIRG}, an instance of hyperbolic random graphs satisfies~\eqref{eq:Euclideanmodel} asymptotically almost surely (over the choice of random weights $\w{}$). Thus, hyperbolic random graphs, which have gained a lot of theoretical and experimental interest during the last years (see, e.g., \cite{BogunaPK10,KrioukovPKVB,gugelmann2012random,friedrich2015cliques}), also are a special case of our general model.

In Section~\ref{sec:distancemodel} we will see that GIRGs can be varied as follows. As before, let $\Space= [0,1]^d$. For $x = (x_1,\ldots,x_d)$ and $y = (y_1,\ldots,y_d) \in \Space$, we define the \emph{minimum component distance} $\|x-y\|_{\min}
:= \min\{|x_i-y_i| \mid 1\leq i \leq d\}$, where the differences $x_i-y_i \in [-1/2,1/2)$ are computed modulo 1, or, equivalently, on the circle. This distance reflects the property of social networks that two
individuals may know each other because they are similar in only one feature
(e.g., they share a hobby), regardless of the differences in other features.
Note that the minimum component distance is not a metric, since there are
$x,y,z\in \Space$ such that $x$ and $y$ are close in one component, $y$ and $z$
are close in one (different) component, but $x$ and $z$ are not close in any
component. Let $V(r)$ be the volume of the ball $B_r(0) := \{x \in \Space \mid \|x\|_{\min} \leq r\}$. Then any $p$ satisfying
\begin{equation*}
 p_{uv} = \Theta\Big( \min\Big\{1, V(\|\x u-\x v\|)^{-\alpha} \cdot \Big( \frac{\w u \w v}{\W}\Big)^{\max\{\alpha,1\}} \Big\} \Big)
\end{equation*}
satisfies conditions~(EP1) and (EP2), so it is a special case of our model.\footnote{These examples also show that our model is incomparable to the (also very general) model of inhomogeneous random graphs studied by Bollob\'as, Janson, and Riordan~\cite{BollobasSR07}. Their model requires sufficiently many long-range edges, so that setting $\alpha > 1$ in (\ref{eq:Euclideanmodel}) yields an edge probability that is not supported by their model. Similarly, the example with the minimum component distance is also not supported by their model.}

\subsection{Results of this paper}
Our results generalize and improve the understanding of Chung-Lu random graphs, hyperbolic random graphs, and other models, as they are special cases of our fairly general model. 
We study the following fundamental structural questions. 

\paragraph{Scale-freeness:} Since we plug in power-law weights $\w{}$, it is not surprising that our model is scale-free. 

\begin{theorem}[Section~\ref{sec:degreesequence}] \label{thm:degseq1}
  Whp\footnote{We say that an event holds \emph{with high probability} (whp) if it holds with probability $1 - n^{-\omega(1)}$.} the degree sequence of our random graph model, not necessarily fulfilling (EP2), follows a power law with exponent~$\beta$ and average degree $\Theta(1)$.
\end{theorem}

\paragraph{Giant component and diameter:}
The connectivity properties of the model for $\beta > 3$ are not very well-behaved, in particular since in this case even threshold hyperbolic random graphs do not possess a giant component of linear size~\cite{bode2014geometrisation}. 
Hence, for connectivity properties we restrict our attention to the regime $2 < \beta < 3$, which holds for most real-world networks~\cite{dorogovtsev2002evolution}. 

\begin{theorem}[Section~\ref{sec:diameter}]\label{thm:diameter}
  Let $2 < \beta < 3$. Whp the largest component of our random graph model has linear size, while all other components have size at most $\log^{O(1)} n$. Moreover, whp the diameter is at most $\log^{O(1)} n$. 
\end{theorem}

\noindent
A better bound of $\Theta(\log n)$ holds for the diameter of Chung-Lu graphs~\cite{chung2004avg} and for hyperbolic random graphs \cite{FriedrichKrohmer15, mueller2017diameter}. It remains an open problem whether the upper bound $O(\log n)$ holds in general for our model. 

\paragraph{Average distance:}
As our main result, we determine the average distance between two randomly chosen nodes in the giant component to be the same as in Chung-Lu random graphs up to a factor $1+o(1)$, showing that the underlying geometry is negligible for this graph parameter. 

\begin{theorem}[Section~\ref{sec:diameter}]\label{thm:avgdist}
 Let $2 < \beta < 3$. Then the average distance of our random graph model is $(2 \pm o(1))\frac{\log \log n}{|\log(\beta-2)|}$ in expectation and with probability $1-o(1)$. 
\end{theorem}

\section{Preliminaries and Notation}\label{sec:notation}

\subsection{Notation}\label{subsec:notation}
For $w\in \R_{\geq 0}$, we use the notation
$  V_{\geq w} := \{v\in V\; |\; \w{v}\geq w\}$ and $V_{\leq w} := \{v\in V\; |\; \w{v}\leq w\}$,
as well as
$
  \W_{\geq w}:=\sum_{v\in V_{\geq w}} \w{v}$ and $\W_{\leq w}:=\sum_{v\in V_{\leq w}}\w{v}$ for sums of weights. Recall that $\wmin = \min\{\w v \mid 1 \le v \le n\}$, similarly we put $\wmax := \max\{\w v \mid 1 \le v \le n\}$ for the maximum weight.
For $u,v\in V$ we write $u\sim v$ if $u$ and $v$ are adjacent, and for
$A,B\subseteq V$ we write $A\sim v$ if there exists $u\in A$ such that $u\sim v$, and we write $A \sim B$ if there exists $v \in B$ such that $A\sim v$. For a vertex
$v\in V$, we denote its neighborhood by $\Gamma(v)$, i.e.\ $\Gamma(v):=\{u\in
V\mid u\sim v\}$. We say that an event holds \emph{with high probability} (whp) if it holds with probability $1 - n^{-\omega(1)}$.

\subsection{Concentration inequalities}\label{sec:tools}
In the proofs we will use the following concentration inequalities.
\begin{theorem}[Chernoff-Hoeffding bound, Theorem 1.1 in \cite{dubhashi2009concentration}] \label{thm:dubhashichernoff}
  Let $X:=\sum_{i\in [n]}X_i$ where for all $i \in [n]$, the random variables $X_i$ are independently
  distributed in $[0,1]$. Then 
\begin{enumerate}[(i)]
\item $\Pr[ X > (1+\eps)\mathbb{E}[X]] \leq \exp\left( -\frac{\eps^2}{3}\mathbb{E}[X] \right)$ for all $0<\eps<1$,
\item $\Pr[ X < (1-\eps)\mathbb{E}[X]] \leq \exp\left( -\frac{\eps^2}{2}\mathbb{E}[X] \right)$ for all $0<\eps<1$, and
\item $\Pr[X>t]\leq 2^{-t}$ for all $t>2e\Ex[X]$.
\end{enumerate}  
\end{theorem}

We will need a concentration inequality which bounds large deviations taking into account some bad event $\mathcal{B}$. 
We start with the following variant of McDiarmid's inequality as given in~\cite{kutinextension}.

\begin{theorem}[Theorem~3.6 in \cite{kutinextension}, slightly simplified] \label{thm:weakconcentration}
Let $X_1,\ldots,X_m$ be independent random variables over $\Omega_1, \ldots, \Omega_m$. Let $X = (X_1,\ldots,X_m)$, $\Omega = \prod_{k=1}^m \Omega_k$ and let $f\colon \Omega \to \mathbb{R}$ be measurable with $0 \le f(\omega) \le M$ for all $\omega \in \Omega$. Let $\mathcal{B} \subseteq \Omega$ such that 
for some $c > 0$ and for all $\omega \in \overline{\mathcal{B}},\omega' \in \Omega$ that differ in only one component we have
$|f(\omega)-f(\omega')| \le c$.
Then for all $t>0$
\begin{equation}\label{eq:twosteps}
\Pr[\vert f(X)-\Ex[f(X)] \vert \ge t] \le 2e^{-\frac{t^2}{8mc^2}}+2\tfrac{mM}{c}\Pr[\mathcal{B}].
\end{equation}
\end{theorem}

Our improved version of this theorem is the following, where in the Lipschitz condition both $\omega$ and $\omega'$ come from the good set $\overline{\mathcal B}$, but we have to consider changes of \emph{two} components at once. Recently, a similar inequality has been proven by Combes \cite{Combes15}, see also \cite{warnke2016method}.

\begin{theorem}
\label{thm:concentration}
Let $X_1,\ldots,X_m$ be independent random variables over $\Omega_1, \ldots, \Omega_m$. Moreover, let $X = (X_1, \ldots, X_m)$, $\Omega = \prod_{k=1}^m \Omega_k$, and let $f\colon \Omega \to \mathbb{R}$ be measurable such that $0 \le f(\omega) \le M$ for all $\omega \in \Omega$. Let $\mathcal{B} \subseteq \Omega$ such that 
for some $c > 0$ and for all $\omega \in \overline{\mathcal{B}},\omega' \in \overline{\mathcal{B}}$ that differ in at most \emph{two} components we have
\begin{equation}\label{eq:Lipschitz}
|f(\omega)-f(\omega')| \le c.
\end{equation}
Then for all $t \ge 2 M \Pr[\mathcal B]$
$$\Pr\big[\vert f(X)-\Ex[f(X)] \vert \ge t \big] \le 2e^{-\frac{t^2}{32mc^2}}+(2\tfrac{mM}{c} + 1)\Pr[\mathcal{B}].$$
\end{theorem}
\begin{proof}
  We say that $\omega, \omega' \in \Omega$ are \emph{neighbors} if they differ in exactly one component $\Omega_k$.
  Given a function $f$ as in the statement, we define a function $f'$ as follows. 
  On $\overline{\mathcal B}$ the functions $f$ and $f'$ coincide. 
  Let $\omega \in \mathcal B$. If $\omega$ has a neighbor $\omega' \in \overline{\mathcal B}$, then choose any such $\omega'$ and set $f'(\omega) := f(\omega')$. Otherwise set $f'(\omega) := f(\omega)$. 
  
  The constructed function $f'$ satisfies the precondition of Theorem~\ref{thm:weakconcentration}. Indeed, let $\omega \in \overline{\mathcal B}$ and $\omega' \in \Omega$ differ in only one position. If $\omega' \in \overline{\mathcal B}$, then since $f'(\omega) = f(\omega)$ and $f'(\omega') = f(\omega')$, and by the assumption on $f$, we obtain $|f'(\omega) - f'(\omega')| \le c$. Otherwise we have $\omega' \in \mathcal B$, and since $\omega'$ has at least one neighbor in $\overline{\mathcal B}$, namely $\omega$, we have $f'(\omega') = f(\omega'')$ for some neighbor $\omega'' \in \overline{\mathcal B}$ of $\omega'$. Note that both $\omega$ and $\omega''$ are in $\overline{\mathcal B}$, and as they are both neighbors of $\omega'$ they differ in at most two components. Thus, by the assumption on $f$ we have
  $ |f'(\omega) - f'(\omega')| = |f(\omega) - f(\omega'')| \le c$. Hence, we can use Theorem~\ref{thm:weakconcentration} on $f'$ and obtain concentration of $f'(X)$. Specifically, since $\Pr[f(X) \ne f'(X)] \le \Pr[\mathcal B]$, and thus $|\Ex[f(X)] - \Ex[f'(X)]| \le M \Pr[\mathcal B]$, we obtain
  \begin{align*} 
    \Pr[|f(X) - \Ex[f(X)]| \ge t] &\le \Pr[\mathcal B] + \Pr[|f'(X) - \Ex[f'(X)]| \ge t - M \Pr[\mathcal B]]  \\
    &\le \Pr[\mathcal B] + \Pr[|f'(X) - \Ex[f'(X)]| \ge t/2],
  \end{align*}
  since $t \ge 2 M \Pr[\mathcal B]$,
  which together with Theorem~\ref{thm:weakconcentration} proves the claim.
\end{proof}

\section{Basic Properties} \label{sec:basic}
In this section, we prove some basic properties of the considered random graph model which repeatedly occur in our proofs. In particular we calculate the expected degree of a vertex and the marginal probability that an edge between two vertices with given weights is present. Let us start by calculating the partial weight sums $\W_{\le w}$ and $\W_{\ge w}$. The values of these sums will follow from the assumptions on power-law weights in Section~\ref{sec:model}.

\begin{lemma} \label{lem:totalweight}
The total weight satisfies $\W=\Theta(n)$. Moreover, for all sufficiently small $\eta > 0$, 
\begin{enumerate}[(i)]
\item $\W_{\ge w} = O( n w^{2-\beta+\eta})$ for all $w \ge \wmin$,
\item $\W_{\ge w} = \Omega( n w^{2-\beta-\eta})$ for all $\wmin \le w \le \barw$,
\item $\W_{\le w} = O(n)$ for all $w$, and
\item $\W_{\le w} = \Omega(n)$ for all $w=\omega(1)$.
\end{enumerate}
\end{lemma}
\begin{proof}
Let $w_1\ge w_0 \ge 0$ be two fixed weights. We start by summing up all vertex-weights between $w_0$ and $w_1$. By Fubini's theorem, we can rewrite this sum as 
\begin{equation}\label{eq:weightsum1}
\sum_{v \in V, w_0 \le \w v \le w_1} \w v = \int_0^{\infty}|V_{\ge \max\{w_0,x\}} \setminus V_{> w_1}| dx = w_0 \cdot |V_{\ge w_0}| + \int_{w_0}^{w_1} |V_{\ge x}| dx- w_1 \cdot |V_{>w_1}| .
\end{equation}
We start with (i) and apply \eqref{eq:weightsum1} with $w_0=w$ and $w_1=\wmax$. Then, the set $V_{>w_1}$ is empty, and we have
$\W_{\ge w}=w \cdot |V_{\ge w}| + \int_{w}^{\wmax} |V_{\ge x}|dx$, thus the assumption (PL2) implies that $\W_{\ge w}$ equals
\[|V_{\ge w}| \cdot w + \int_w^{\infty} |V_{\ge x}|dx = O\Big(nw^{2-\beta+\eta}+\int_w^{\infty} nx^{1-\beta+\eta}dx\Big)=O\Big(n w^{2-\beta+\eta}\Big).\]
For (ii) we similarly obtain
$$\W_{\ge w} = \Omega\Big(nw^{2-\beta-\eta}+\int_w^{\wmax}nx^{1-\beta-\eta}dx\Big)=\Omega\Big(n w^{2-\beta-\eta}\Big).$$
For (iii), we see that if $w < \wmin$, then clearly $\W_{\le w}=0$. Otherwise, Equation~\eqref{eq:weightsum1} with $w_0=\wmin$ and $w_1=w$ implies
$$\W_{\le w}=|V_{\ge \wmin}| \cdot \wmin  + \int_{\wmin}^w |V_{\ge x}| dx - |V_{> w}| \cdot w  \le n\wmin + O\Big(\int_{\wmin}^w nx^{1-\beta+\eta} dx\Big)=O(n),$$
and for (iv) we obtain
$$\W_{\le w} \ge \int_{\wmin}^w |V_{\ge x}| dx- |V_{> w}| \cdot w = \Omega\Big(\int_{\wmin}^w nx^{1-\beta-\eta} dx\Big)-O\Big(nw^{2-\beta+\eta}\Big) =\Omega(n)-o(n)=\Omega(n).$$
In particular, with the choice $w=\wmax$ the property $\W=\Theta(n)$ follows from (iii) and (iv).
\end{proof}

Next we consider the marginal edge probability $\Pr[u \sim v]$ of two vertices $u$, $v$ with weights $\w{u}$, $\w{v}$. For a fixed position $\x u \in \Space$, we already know this probability by (EP1).

\begin{lemma} \label{lem:marginal}
  Let $u \in [n]$ and let $\x u \in \Space$ be any fixed position. Then all edges $\{u,v\}$, $u \ne v$, are independently present with probability
\[\Pr[u \sim v \mid \x{u}] = \Theta\big(\Pr[u \sim v]\big) =   \Theta\Big(\min\Big\{1,\frac{\w{u} \w{v}}{\W} \Big\}\Big).
  \end{equation*}
\end{lemma}

\begin{proof}Let $u,v \in [n]$. Then by (EP1), it follows directly
$$\Pr[u \sim v] = \Ex_{\x u}\big[\Pr_{\x v}[u \sim v \mid \x{u}]\big]=\Ex_{\x u}\Big[\Theta\Big(\min\Big\{1,\frac{\w{u}\w{v}}{\W}\Big\}\Big)\Big]  =\Theta\Big(\min\Big\{1,\frac{\w{u} \w{v}}{\W} \Big\}\Big).$$

Furthermore, for every fixed $\x u \in \Space$ the edges incident to $u$ are independently present with probability $\Pr_{\x v}[u \sim v \mid \x{u}]$, as the event ``$u \sim v$'' only depends on $\x{v}$, and an independent random choice for the edge $uv$ (after fixing $\x u$).
\end{proof}

The following lemma shows that the expected degree of a vertex is of the same order as the weight of the vertex, thus we can interpret a given weight sequence $\w{}$ as a sequence of expected degrees.

\begin{lemma}\label{lem:expecteddegree}
For any $v \in [n]$ we have $\Ex[\deg(v)]=\Theta(\w{v})$.
\end{lemma}
\begin{proof}
Let $v$ be any vertex. We estimate the expected degree both from below and above. By Lemma~\ref{lem:marginal}, the expected degree of $v$ is at most
$$
\sum_{u \neq v} \Pr[u \sim v] = \Theta\Big(\sum_{u \neq v} \min\Big\{1,\frac{\w{u}\w{v}}{\W}\Big\}\Big)=O\Big(\sum_{u \in V} \frac{\w{u}\w{v}}{\W}\Big) = O\Big(\frac{\w{v}}{\W}\sum_{u\in V} \w{u}\Big) = O(\w{v}).$$
For the lower bound, $\Pr[u \sim v]=\Theta(\frac{\w{u}\w{v}}{\W})$ holds for all $\w{u} \leq \frac{\W}{\w{v}}$. We set $w' := \frac{\W}{\w{v}}$ and observe that $w'=\omega(1)$. Using Lemma~\ref{lem:totalweight}, we obtain
$$\Ex[\deg(v)]  \ge \sum_{u \neq v, u \in V_{\le w'}} \Pr[u \sim v] = \Omega\Big(\frac{\w{v}}{\W} \W_{\leq w'}\Big) = \Omega(\w{v}).$$
\end{proof}

As the expected degree of a vertex is roughly the same as its weight, it is no surprise that whp the degrees of all vertices with sufficiently large weight are concentrated around the expected value. The following lemma gives a precise statement.

\begin{lemma} \label{lem:largevertices}
The following properties hold whp.
\begin{enumerate}[(i)]
\item $\deg(v) = O(\w{v} + \log^2 n)$ for all $v \in [n]$.
\item $\deg(v)= (1+o(1))\Ex[\deg(v)]= \Theta(\w{v})$ for all $v \in V_{\ge \omega(\log^2 n)}$.
\item $\sum_{v \in V_{\ge w}} \deg(v) = \Theta(\W_{\ge w})$ for all $w=\omega(\log^2 n)$.
\end{enumerate}
\end{lemma}
\begin{proof}
Let $v \in V$ with fixed position $\x{v} \in \Space$ and let $\mu :=  \Ex[\deg(v) \mid \x{v}]=\Theta(\w{v})$. By definition of the model, conditioned on the position $\x{v}$ the degree of $v$ is a sum of independent Bernoulli random variables. By Lemma~\ref{lem:expecteddegree} there exists a constant $C$ such that $2e\mu < C \log^2 n$ holds for all vertices $v \in V_{\le \log^2 n}$ and all positions $\x{v} \in \Space$. Thus, if $v \in V_{\le \log^2 n}$, we apply a Chernoff bound (Theorem~\ref{thm:dubhashichernoff}.(iii)), and obtain $\Pr[\deg(v) > C \log^2 n] \le 2^{-C \log^2 n}=n^{-\omega(1)}$. If $v \in V_{\ge \log^2 n}$, we similarly obtain $\Pr[\deg(v) > 3\mu/2] \le e^{-\Theta(\mu)}=n^{-\omega(1)}$ and $\mu=\Theta(\w{v})$ by Lemma~\ref{lem:expecteddegree}. Then (i) follows by applying a union bound over all vertices.

For (ii), let $v \in V$ such that $\w{v}=\omega(\log^2 n)$, let $\mu$ be as defined above and put $\eps=\frac{\log n}{\sqrt{\mu}}=o(1)$. Thus by the Chernoff bound, 
\[\Pr\big[|\deg(v)-\mu| > \varepsilon \cdot\mu\big] \le e^{-\Theta(\varepsilon^2 \cdot\mu)}=n^{-\omega(1)},\]
and we obtain (ii) by applying Lemma~\ref{lem:expecteddegree} and a union bound over all such vertices. Finally, from (ii) we infer $\sum_{v \in V_{\ge w}} \deg(v) = \sum_{v \in V_{\ge w}} \Theta(\w{v}) = \Theta(\W_{\ge w})$ for all $w = \omega(\log^2 n)$, which shows (iii).
\end{proof}

We conclude this section by proving that if we sample the weights randomly from an appropriate distribution, then almost surely the resulting weights satisfy our conditions on power-law weights.

\begin{lemma} \label{lem:sampleweights}
Let $\wmin$ be a strictly positive constant, let $F=F_n: \R \rightarrow [0,1]$ be non-decreasing such that $F(z)=0$ for all $z \le \wmin$, and $F(z)=1-\Theta(z^{1-\beta})$ for all $z \ge \wmin$. Suppose that for every vertex $v \in [n]$, we choose the weight $\w{v}$ independently according to the cumulative probability distribution $F$. Then with $\barw = (n/\log^2 n)^{1/(\beta-1)}$, the resulting weight vector $\w{}$ satisfies deterministically (PL1), whp the lower bound of (PL2), and for all $\eta=\eta(n)=\omega(\log \log n/\log n)$ with probability $1-n^{-\Omega(\eta)}$ the upper bound of (PL2) .
\end{lemma}

In particular, this lemma proves that for all small constants $\eta>0$, with probability $1-n^{-\Omega(1)}$ (PL1) and (PL2) are fulfilled for weights sampled according to $F(\cdot)$. Moreover, it follows that any property which holds with probability $1-q$ for weights satisfying (PL1) and (PL2) also holds in a model of sampled weights with probability at least $1-q-n^{-\Omega(1)}$. However, we claim without proof that all our results hold with the original probability in a model of sampled weights.

\begin{proof}[Proof of Lemma~\ref{lem:sampleweights}]

Condition (PL1) is fulfilled by definition of $F$, and we only need to prove (PL2). 
For all $z>\wmin$, denote by $Y_z$ the number of vertices with weight at least $z$ and observe that
\begin{equation} \label{eq:expsampleweights}
\Ex[Y_z] = n (1-F(z))= \Theta(n z^{1-\beta}).
\end{equation}

Let us first consider the case $z \in [\wmin,\barw]$. For all $z$ in this range we have $\Ex[Y_z] = \Omega(\log^2 n)$, so for any $z \in [\wmin,\barw]$ the Chernoff bound (Theorem~\ref{thm:dubhashichernoff}.(i) and (ii)) yields
$$ \Pr[|Y_z-\Ex[Y_z]| \ge 0.5 \Ex[Y_z]] \le \exp(-\Omega(\Ex[Y_z])) = n^{-\Omega(\log n)}. $$
Note that $Y_z$ is always an integer and at most $n$. Clearly, $\Ex[Y_z]$ is decreasing.
Hence, we can assume without loss of generality that either $z \in \{\wmin,\barw\}$ or $0.5 \Ex[Y_z]$ or $1.5 \Ex[Y_z]$ is an integer, because if 
$0.5 \Ex[Y_z]< Y_z < 1.5 \Ex[Y_z]$
holds for these values of~$z$, then it \emph{must} hold for all other $z \in [\wmin,\barw]$ as well. Thus, we can restrict $z$ to a set of size $O(n)$. This allows us to take a union bound, and it follows that with probability $1-n^{-\omega(1)}$, $Y_z=\Theta(nz^{1-\beta})$ holds for all $z \in [\wmin,\barw]$. In this case, all $z$ in our range satisfy both the lower and upper bound of (PL2) even for $\eta=0$. In particular, this proves that with probability $1-n^{-\omega(1)}$, the lower bound of (PL2) holds for all $\eta\ge 0$.

It only remains the upper bound of (PL2) for $z\ge \barw$. Let $\eta=\eta(n)=\omega(\log\log n/\log n)$ and $z \ge \barw$. By Markov's inequality and \eqref{eq:expsampleweights},
\begin{equation}\label{eq:Markovsampleweights}
 \Pr[Y_z \ge nz^{1-\beta+\eta}] \le \Pr[Y_z \ge\Omega(z^{\eta})\Ex[Y_z]] \le O(z^{-\eta})  \le n^{-\Omega(\eta)}.
\end{equation}

By the same argument as above, we can restrict $z$ to $\barw$ and values where the intended bound $\Omega(z^{\eta})\Ex[Y_z]$ is integral, which happens only for $O(\log^2 n)$ values of $z$ above $\barw$. Hence, we can use the union bound to obtain error probability 
\[O(n^{-\Omega(\eta)} \log^2 n)=n^{-\Omega(\eta)},\]
since $\eta(n)=\omega(\log\log n/\log n)$.
In particular it also follows that with probability $1-n^{-\Omega(\eta)}$, the maximum weight satisfies $\wmax \le n^{1/(\beta-1-\eta)}$.

\end{proof}

\section{Giant Component, Diameter, and Average Distance}
\label{sec:diameter}

Throughout this section we assume $2<\beta <3$. Under this assumption we prove that whp the general model has a giant component with diameter at most $(\log n)^{O(1)}$, and that all other components are only of polylogarithmic size. 
We further show that the average distance of any two vertices in the giant is $(2 + o(1))\log \log n/|\log(\beta -2)|$ in expectation and with probability $1-o(1)$. The same formula has been known to hold for various graph models, including Chung-Lu~\cite{chung2004avg} and hyperbolic random graphs~\cite{abdullah2015typical}. The lower bound follows from the first moment method on the number of paths of different types. Note that the probability that a fixed path $P=(v_1,\ldots,v_k)$ exists in our model is the same as in Chung-Lu random graphs, since the marginal probability of the event $v_i \sim v_{i+1}$ conditioned on the positions of $v_1,\ldots v_i$ is $\Theta(\min\{1,\w{v_i}\w{v_{i+1}}/\W\})$, as in the Chung-Lu model. In particular, the expected number of paths coincides for both models (save the factors coming from the $\Theta(\cdot)$-notation). Not surprisingly, the lower bound for the expected average distance follows from general statements on power-law graphs, bounding the expected number of too short paths by $o(1)$, cf.~\cite[Theorem 2]{dereich2012typical}. The main contribution of this section is to prove a matching upper bound for the average distance. 

The proof-strategy is as follows. We first prove that whp for every vertex of weight at least $(\log n)^C$ there exists an ultra-short path to the ``heavy core'', which has diameter $o(\log\log n)$ and contains the vertices of highest weight. Afterwards, we show that a random low-weight vertex has a large probability to connect to a vertex of weight at least $(\log n)^C$ within a small number of steps. The statement is formalized below as the ``bulk lemma'' (Lemma~\ref{lem:path-to-core3}). This lemma is the crucial step of the main proof and new compared to previous studies of Chung-Lu random graphs and similar models. It contains a delicate analysis of the $k$-hop-neighborhood of a random vertex, restricted to small weights. Thereby, the underlying geometry is used implicitly, in order to make the argument applicable for the fairly general model that we study.

In the whole section let $G$ be a graph sampled from our model. We start by considering the subgraph induced by the \emph{heavy vertices} $\barV := V_{\geq \barw}$, where $\barw$ is given by the definition of power-law weights, see condition (PL2). We call the induced subgraph $\barG := G[\barV]$ the \emph{heavy core}.

\begin{lemma}[Heavy core]\label{lem:diameter-core}
With high probability $\bar G$ is connected and has diameter $o(\log \log n)$.
\end{lemma}
\begin{proof}
Let $\bar n$ be the number of vertices in the heavy core, and let $\eta>0$ be small enough. Since $\barw \leq n^{(1-\Omega(1))/(\beta-1)}$, we may bound $\bar n = \Omega(n\barw^{1-\beta -\eta}) = n^{\Omega(1)}$. By (EP2), the connection probability for any heavy vertices $u,v$, regardless of their position, is at least
\begin{equation*}
p_{uv}(\x u, \x v) \geq \Big(\frac{n}{\barw^{\beta-1+\eta}}\Big) ^{-1+\omega(1/\log\log n)} \ge \bar n^{-1+\omega(1/\log\log n)}.
\end{equation*}
Therefore, the diameter of the heavy core is at most the diameter of an Erd\H{o}s-R\'enyi random graph $G(\bar n,p)$, with $p= \bar n^{-1+\omega(1/\log\log n)}$. However, with probability $1-\bar n^{-\omega(1)}$ this diameter is $\Theta(\log \bar n/\log(p\bar n))=o(\log \log n)$ \cite{Draief2010}. Since $\bar n = n^{\Omega(1)}$, this proves the lemma.
\end{proof}

Next we show that if we start at a vertex of weight $w$, going greedily to neighbors of largest weight yields a short path to the heavy core with a probability that approaches $1$ as $w$ increases.
\begin{lemma}[Greedy path]\label{lem:path-to-core}\leavevmode
\begin{enumerate}[(i)]
\item Let $0<\eps<1$ and let $v$ be a vertex of weight $2 \le w < \barw$. Then with probability at least 
$1-O\left(\exp\left(-w^{\Omega(\eps)}\right)\right)$ 
there exists a weight-increasing path of length at most $(1+\eps)\frac{\log \log n}{|\log(\beta-2)|}$ from $v$ to the heavy core.
\item For every $\eps>0$ there exists a constant $C=C(\eps)>0$ such that whp for all $v \in V_{\geq (\log n)^{C}}$ there exists a weight-increasing path of length at most $(1+\eps)\frac{\log \log n}{|\log(\beta-2)|}$ from $v$ to the heavy core.
\item Whp there are $\Omega(n)$ vertices in the same component as the heavy core.
\end{enumerate}
\end{lemma}
\begin{proof}
Let $0<\eps<1$, let $v$ be a vertex of weight $2 \le w \le \barw$, and let 
\[\tau=\tau(\varepsilon) := (\beta-2)^{-1/(1+\eps/2)}.\] 
Note that $1 < \tau < 1/(\beta-2)$, and that $1/\log \tau = (1+\eps/2)/|\log(\beta-2)|$. Moreover, we define an increasing weight sequence $w_0, w_1, \ldots, w_{i_{\max}} := \barw$ such that for all $1 \le i \le i_{\max}$ it holds $w_i := w_{i-1}^{\tau}$, and such that $w_0 \le \w v < w_1$. For all $i < i_{\max}$ we put $V_i := V_{\ge w_i} \setminus V_{\ge w_{i+1}}$. Furthermore, we put $V_{i_{\max}} := \barV$ and $v_0 := v$. We will show that with sufficiently high probability, for all $0 \le i <  i_{\max}$ the vertex $v_{i}$ has at least one neighbor $v_{i+1} \in V_{i+1}$. Note that
$$i_{\max} = \lceil\log_\tau\left(\log \barw/\log w\right)\rceil,$$ so this implies that there is a path from $v$ to the heavy core of length at most 
\[i_{\max} \leq (1+\eps/2)\frac{\log \log n}{|\log(\beta-2)|} + 1 \le (1+\eps) \frac{\log \log n }{|\log(\beta-2)|},\] for sufficiently large $n$, and thus proves statement (i).

Let $0\leq i< i_{\max}$ and assume by induction that there exists a weight-increasing path from $v_0$ to $v_{i}$ where $v_{i} \in V_i$. Note that this event only depends on the random graph induced by the vertex set $V_{< w_{i+1}}$. We want to verify that $v_{i}$ connects to at least one vertex $v_{i+1} \in V_{i+1}$. First, observe that by condition~(PL2), each layer $V_i$ contains at least $\Omega(n w_i^{1-\beta-\eta})$ and at most $O(n w_i^{1-\beta+\eta})$ vertices. Next, by condition~(EP1) the edges from $v_{i}$ to vertices $v$ with $v\in V_{i+1}$, are independently present with probability $\Omega(\min\{\w{v}w_{i}/\W,1\})$, respectively. If $w_{i}w_{i+1} \geq \W$, this probability is $\Omega(1)$. However, then $w_{i} \ge n^{1/(1+\tau)}$ and we deduce $|V_{i+1}|=n^{\Omega(1)}$. In this case, the probability that $v_{i}$ connects to at least one vertex of the next weight layer is 
$$1-\exp(-n^{\Omega(1)}) = 1-\exp(-w_i^{\Omega(\eps)}).$$ So assume $w_{i}w_{i+1} < \W$, where we can lower-bound the edge probability by $\Omega(w_{i} w_{i+1} / \W)$. 
Thus, for any $\eta>0$ the probability that $v_{i-1}$ does \emph{not} connect to a vertex in $V_{i+1}$ is at most
$$p_i := \prod_{v \in V,\, \w{v}\geq w_{i+1}} \Big(1-\Omega\Big(\frac{w_i w_{i+1}}{\W}\Big)\Big)   \leq  \exp(-\Omega\Big(\frac{w_i w_{i+1}}\W \cdot |V_{i+1}|\Big)  \Big)  \leq \exp\Big(-\Omega\big(w_i w_{i+1}^{2-\beta-\eta}\big) \Big),$$
where we used Lemma~\ref{lem:totalweight} in the last step. Since $w_{i+1} \le w_i^\tau$, we obtain
\[p_i \le \exp\big(-\Omega\big(w_i^{1-\tau(\beta-2+\eta)}\big) \big).\]
Note that as $\tau < 1/(\beta-2)$, the exponent of $w_i$ in this expression is positive for sufficiently small $\eta>0$. More precisely, we have
$1 - \tau(\beta-2) = 1 - (\beta-2)^{\eps/(2+\eps)} = \Omega(\eps)$,
and thus for $\eta > 0$ sufficiently small compared to $\eps$ we have
\begin{equation} \label{eq:greedypath}
  p_i \le \exp\big( -w_i^{\Omega(\eps)} \big). 
\end{equation}
By the union bound, the probability that for every $0\leq i < i_{\max}$ the vertex $v_i$ has a neighbor in the next weight layer is at least $1-\sum_i\exp(- w_i^{\Omega(\eps)} ) = 1-\exp(-w^{\Omega(\eps)} )$, which proves the first claim.

For the second statement, let $C=C(\eps)=\Omega(1/\eps)$ with sufficiently large hidden constant. If a vertex $v$ has weight at least $(\log n)^C$ then the error probability estimated above is at least $1-e^{-(\log n)^{\Omega(1)}} = 1-n^{-\omega(1)}$. The claim now follows from a union bound over all vertices of weight at least $(\log n)^C$. 

For the size of the giant component, we apply the same arguments as before in the proof of (i) for $w=2$. Let $\eps>0$ be sufficiently small, let $\eta >0$ sufficiently small compared to $\eps$, and consider the same system of weight layers $V_i$ as before. Let $i \ge 0$ such that $w_i \le (\log n)^C$, where $C$ is the constant $C(\eps)$ given by the proof of statement (ii). For every $v \in V_i$, let $\Gamma_i(v) := \{u \in V_{i+1} \mid v \sim u\}$, and let $E_i(v) := \Ex[|\Gamma_i(v)|]$. Moreover, let $\gamma := \tau(2-\beta-\eta)+1>0$. Then for every $v \in V_i$, by (EP1),
\[E_i(v) \geq \Omega\Big(n w_{i+1}^{1-\beta-\eta} \cdot \frac{w_i w_{i+1}}{\W}\Big) \geq \Omega(w_{i+1}^{2-\beta-\eta}w_i) \geq \Omega(w_i^{\tau(2-\beta-\eta)+1})\geq \Omega(w_i^{\gamma}).\]
As this lower bound is independent of $v$, we also have 
\[E_i := \min_{v\in V_i} E_i(v) = \Omega(w_i^{\gamma}).\]
Let $\lambda := \min\{\gamma, \frac{1}{2C\tau}\}$. Furthermore put $B_i := \{v\in V_i \mid |\Gamma_i(v)| \leq E_i/2\}$. This set will play the role of ``bad'' vertices. 

\begin{claim}\label{cla:nrbadvertices}
There is a constant $c>0$ such that for all $i \ge 0$ with $w_i \le (\log n)^C$, whp it holds $|B_i| \leq 2\exp(-cw_i^{\lambda})\cdot |V_i|$.
\end{claim}
We postpone the proof of Claim~\ref{cla:nrbadvertices} (and Claim~\ref{cla:crossingedges} below) until we have finished the main argument. We uncover the sets $V_i$ one by one, starting with the largest weights. Let $\delta>0$ be so small that $\tau(\lambda -\delta) > \lambda$. Note that when applying Claim~\ref{cla:nrbadvertices}, we may replace the factor $2$ by any other factor $D_1\geq 2$ without violating the statement of the claim. We will show by induction that if $D_1 = O(1)$ is sufficiently large, then whp the fraction of vertices in $V_i$ with a weight-increasing path to the inner core is at least $1-D_1\exp(-c w_i^{\lambda -\delta})$. Note that for any $i_0 = i_0(c) = O(1)$ the statement is trivial for all $i \leq i_0$, if we choose $D_1= D_1(i_0,c)$ sufficiently large. Also, if $w_i \geq (\log n)^C$ then we already know that whp all vertices in $V_i$ are connected to the inner core with weight-increasing paths. For the remaining values of $i$, denote by $V_{i}'$ the set of vertices in $V_{i}$ for which there is no weight-increasing path to the inner core that uses exactly one vertex per layer. Furthermore, let $\Lambda_i$ be the set those $D_1\exp(-c w_{i-1}^{\lambda -\delta})|V_{i}|$ vertices in $V_i$ with the smallest neighborhood in $V_{i+1} \setminus V_{i+1}'$ (where we break ties according to some previously fixed order). We then use the following claim.

\begin{claim}\label{cla:crossingedges}
There exists $D_2>0$ such that whp, for all $i \ge i_0$ with $w_i \le (\log n)^C$, it holds
\[|E(V_{i},\Lambda_{i+1})| \leq D_1\exp(-cw_{i+1}^{\lambda -\delta})\cdot |V_i| \cdot E_{i} \cdot w_{i}^{D_2}.\]
\end{claim}

Consider some $i_0<i \leq i_{\max}$ such that $w_i \le (\log n)^C$, and assume by induction that for sufficiently many vertices of $V_{i+1}$ there is a weight-increasing path to the inner core, that is, $|V_{i+1}'| \leq D_1\exp(-cw_{i+1}^{\lambda -\delta})\cdot |V_{i+1}|$. By construction, this implies $V_{i+1}' \subseteq \Lambda_{i+1}$. Now we consider $B_i' := \{v \in V_{i} \mid |E(\{v\}, V_{i+1}')| \geq E_{i}/2\}$. If the low-probability event of Claim~\ref{cla:crossingedges} does not occur, using $w_{i+1}^{\lambda-\delta} = w_i^{\lambda+\Omega(1)}$ it follows
\begin{align}\label{eq:giantbad2}
|B_i'| & \leq \frac{2|E(V_{i},V_{i+1}')|}{E_i}\leq 2D_1\exp(-cw_{i+1}^{\lambda -\delta})\cdot |V_i| \cdot w_{i}^{D_2} \leq D_1 \exp(-cw_{i}^{\lambda})\cdot |V_i| ,
\end{align}
provided that $i_0 = i_0(c)$ (and thus, $w_{i_0}$) is a sufficiently large constant. It remains to observe that every vertex in $V_i \setminus (B_i \cup B_i')$ has at least one edge into $V_{i+1} \setminus V_{i+1}'$. Since the latter vertices are all connected to the inner core, we have at least $|V_i|-|B_i|-|B_{i}'|$ vertices in $V_i$ that are connected to the inner core. By Claim~\ref{cla:nrbadvertices} and Equation~\eqref{eq:giantbad2}, whp both $B_i$ and $B_i'$ have size at most $D_1\exp(-cw_{i}^{\lambda})|V_i|$, so together they have size at most $D_1\exp(-cw_{i}^{\lambda-\delta})|V_i|$, for all $i \geq i_0$ where $i_0$ is sufficiently large. This concludes the induction modulo Claims~\ref{cla:nrbadvertices} and~\ref{cla:crossingedges}. The existence of the giant component now follows because whp a constant fraction of $V_{i_0}$ is connected to the inner core, and $V_{i_0}$ has linear size by~(PL2).
\end{proof}

\begin{proof}[Proof of Claim~\ref{cla:nrbadvertices}]
Let $i\ge 0$ such that $w_i \le (\log n)^C$. For a single $v \in V_i$, the events ``$v \sim u$'' are independent for all $u \in V_{i+1}$. So by the Chernoff bound (Theorem~\ref{thm:dubhashichernoff}), there is a constant $c>0$ such that $\Pr[v \in B_i] \leq \exp(-cw_i^{\gamma})$ and $\Ex[|B_i|] \leq \exp(-cw_i^{\gamma})|V_i|$. Let $G_i$ be the subgraph induced by $V_i$ and $V_{i+1}$ and observe that the size of $B_i$ only depends on $G_i$. In order to prove concentration of $|B_i|$ we will use Theorem~\ref{thm:concentration}. For this, we need to argue that the considered probability space $\Omega$ is a product of independent random variables. Recall that two different random processes are applied to create the geometric graph. First, we choose the positions $\x v \in \Space$ independently at random. Afterwards, every edge $\{u,v\}$ is inserted with some probability $p_{uv}$. So far, these random variables are \emph{not} independent.

W.l.o.g.\ assume that the vertices are sorted by weights in decreasing order. For every vertex $u \in V_i \cup V_{i+1}$ we first have the random variable $\x v$ for its position. Now, for each $u \in V_i \cup V_{i+1}$ we introduce a second, independent random variable $Y_u := (Y_u^1, \ldots, Y_u^{u-1})$, where each $Y_u^v$ is a real number chosen independently and uniformly at random from the interval $[0,1]$. Then for $v<u$, we include the edge $\{u,v\}$ in the graph if and only if
\[p_{uv} > Y_u^v.\]
We observe that indeed this implies $\Pr[u \sim v \mid \x{u},\x{v}]=p_{uv}(\x{u},\x{v})$, as desired. Furthermore, all random variables of the set $\cup_{u \in V_i \cup V_{i+1}} \{\x{u} \cup Y_u\}$ are independent, and together define a product probability space that is equivalent to $\Omega$ and consists of $2(|V_i|+|V_{i+1}|)$ independent coordinates. Formally, every $\omega \in \Omega$ defines a graph $G_i(\omega)$.

We study the bad event $\mathcal{B}$ that there exists a vertex $v \in V_i \cup V_{i+1}$ with degree larger than $(\log n)^{2C\tau^2}$ in $G_{i}$. By Lemma~\ref{lem:largevertices} we have $\Pr[\mathcal{B}] = n^{-\omega(1)}$, since $w_i \le (\log n)^C$ and therefore $\w v \le (\log n)^{C\tau^2}$ for all $v \in V_i \cup V_{i+1}$. Let $\omega,\omega' \in \overline{\mathcal{B}}$ such that they differ in at most two coordinates of our product probability space $\Omega$. We observe that changing one coordinate $\x{u}$ or $Y_u$ can only influence the degrees of $u$ itself and of the vertices that are neighbors of $u$ before or after the coordinate change. 
Therefore, $\left| |B_i(\omega)| - |B_i(\omega')| \right| \le (\log n)^{O(1)}$. We pick $t = \exp(-cw_i^{\lambda})\cdot |V_i|$ and observe that $w_i^{\lambda} \le (\log n)^{1/2}$ by our choice of $\lambda$. Then Theorem~\ref{thm:concentration} implies
\begin{align*}
\Pr[|B_i|-\Ex[|B_i|] \ge t] &\le 2\exp\Big(-\frac{t^2}{128|V_i|(\log n)^{O(1)}}\Big)+n^{O(1)}\Pr[\mathcal{B}]\\ &\le 2\exp\Big(-\frac{e^{-2cw_i^\lambda}|V_i|}{128 (\log n)^{O(1)}}\Big)+n^{O(1)}\Pr[\mathcal{B}] = n^{-\omega(1)}.
\end{align*}
Hence, whp we have $|B_i| \leq \Ex[|B_i|] + t \leq (\exp(-cw_i^{\gamma})+\exp(-cw_i^{\lambda}))\cdot |V_i|$. For $i$, the statement now follows since $\lambda < \gamma$ and $w_i >1$, and then the proof of the claim is finished by a union bound over all $O(\log\log n)$ choices of $i$.
\end{proof}

\begin{proof}[Proof of Claim~\ref{cla:crossingedges}]
Let $i\ge i_0$ such that $w_i \le (\log n)^C$. We assume that the subgraph induced by $V_{\ge w_{i+2}}$ is given, and now we uncover $V_i$ and $V_{i+1}$ to obtain the subgraph induced by $V_{\ge w_i}$. Similarly as in the proof of Claim~\ref{cla:nrbadvertices} we can assume that this probability space $\Omega$ is a product probability space with $2(|V_i|+|V_{i+1}|)$ coordinates. Recall that $G_i$ denotes the subgraph induced by $V_i \cup V_{i+1}$. We consider the same bad event $\mathcal{B}$ as in the proof of Claim~\ref{cla:nrbadvertices}, i.e., $\mathcal{B}$ denotes the event that the maximum degree in $G_{i}$ is larger than $(\log n)^{2C\tau^2}$. Note that $\mathcal{B}$ is independent of $V_{\ge w_{i+2}}$, so indeed Lemma~\ref{lem:largevertices} can be again applied to deduce $\Pr[\mathcal{B}]=n^{-\omega(1)}$. 

Let $Z_i := |E(V_i, \Lambda_{i+1})|$, and let $\omega, \omega' \in \overline{\mathcal{B}}$ such that they differ in at most two coordinates of $\Omega$. If we change a coordinate of $\Omega$ that stems from a vertex $v \in V_i$, under $\overline{\mathcal{B}}$ the influence on $Z_i$ is at most $(\log n)^{O(1)}$. If a coordinate belonging to a vertex of $V_{i+1}$ is changed, this may result in a different set $\Lambda_{i+1}$. However, the symmetric difference between the old $\Lambda_{i+1}$ and the new $\Lambda_{i+1}$ is at most two (as $\Lambda_{i+1}$ is defined via a fixed ordering), and under $\overline{\mathcal{B}}$ the influence on $Z_i$ is again upper-bounded by $(\log n)^{O(1)}$. Finally, the same is true if the set $\Lambda_{i+1}$ does not change. Repeating the argument for the second coordinate change, we conclude that $|Z_i(\omega)-Z_i(\omega')| \le (\log n)^{O(1)}$.

Next, we want to upper-bound $\Ex[Z_i]$. First, we uncover $V_{i+1}$ to obtain the subgraph induced by $V_{\ge w_{i+1}}$. Then the set $\Lambda_{i+1}$ is determined. In a second step, we uncover $V_i$. By (EP1) and linearity of expectation, we deduce
\begin{align*}
\Ex[Z_i] & \leq |\Lambda_{i+1}| \cdot |V_{i}| \cdot O\Big(\frac{w_{i+1}^{1+\eta} w_{i}^{1+\eta}}{\W}\Big)\\ 
& = O\Big(D_1\exp(-cw_{i+1}^{\lambda-\delta})n w_{i+1}^{1-\beta+\eta}\cdot |V_i|\cdot \frac{w_{i+1}^{1+\eta} w_{i}^{1+\eta}}{\W}\Big) \\
& \leq D_1\exp(-cw_{i+1}^{\lambda-\delta}) \cdot |V_i| \cdot O\big(w_{i+1}^{2-\beta+2\eta}w_{i}^{1+\eta}\big) \\
& \leq D_1\exp(-cw_{i+1}^{\lambda-\delta})\cdot  |V_i|\cdot  O\big(w_{i}^{\tau(2-\beta)+1+\eta(2\tau+1)}\big) \\
& \leq D_1\exp(-cw_{i+1}^{\lambda-\delta})\cdot  |V_i|\cdot E_i \cdot O\big(w_{i}^{\eta(3\tau+1)}\big).
\end{align*}
Since we assumed $w_i \geq 2$, we may upper-bound the $O(\cdot)$-term by $0.5 w_i^{D_2}$ for a sufficiently large $D_2>0$.

Now we can apply Theorem~\ref{thm:concentration} with $t = 0.5D_1\exp(-cw_{i+1}^{\lambda-\delta})\cdot |V_i| \cdot E_{i} \cdot w_{i}^{D_2}$. Using $\lambda\tau C \le \frac12$, it follows similarly as in the proof of Claim~\ref{cla:nrbadvertices} that $\Pr\big[|Z_i - \Ex[Z_i]| \ge t \big] = n^{-\omega(1)}$, and we conclude that with probability $1-n^{-\omega(1)}$ it holds
\[|Z_i| \le \Ex[|Z_i|]+ t \le D_1\exp(-cw_{i+1}^{\lambda-\delta})\cdot |V_i| \cdot E_{i} \cdot w_{i}^{D_2}.\]
Now the claim follows by a union bound over all $O(\log\log n)$ choices of $i$.
\end{proof}

By Lemma~\ref{lem:path-to-core}~(ii), whp every vertex of weight at least $(\log n)^C$ has small distance from the heavy core. It remains to show that every vertex in the giant component has a large probability to connect to such a high-weight vertex in a small number of steps. The next lemma shows that the more vertices of small weight we have in the neighborhood of a vertex, the more likely it is that there is an edge from the neighborhood to a vertex of large weight.

\begin{lemma}[Bulk lemma]\label{lem:path-to-core3}
Let $\eps >0$. Let $\wmin \leq w \leq \barw$ be a weight, and let $k \geq \max\{2, w^{\beta+\eps}\}$ be an integer. For a vertex $v \in V_{<w}$, let $N_v$ be the set of all vertices within distance at most $k$ of $v$ in the graph $G_{< w}$. Then for a random vertex $v\in V_{<w}$,
\[
\Pr\big[\dist(v,V_{\geq w}) > k \text{ and } |N_v| \geq k\big] \leq O\big(e^{-w^{\Omega(1)}}\big).
\]
\end{lemma}
\begin{proof}Before starting with the formal proof, let us sketch some of the main ideas. We first uncover the graph $G_{<w}$ induced by vertices of weights less than $w$. For a fixed vertex $v \in V_{<w}$ and a vertex $u \in V_{\ge w}$, we know (a lower bound on) $\Ex[|N_v \cap \Gamma(u)| \mid G_{<w}]$ by Lemma~\ref{lem:marginal}. This can only come from two cases: either $N_v$ is such that $\Pr[u \sim N_v]$ is large, or we have a relatively large probability that $u$ connects to many vertices of $N_v$ at the same time. In a geometric setting like hyperbolic random graphs or GIRGs (see Section~\ref{sec:distancemodel}), intuitively the first case occurs if the vertices in $N_v$ are spread out in the geometric space, while the second case occurs if the vertices in $N_v$ form a bulk. In the first case, it is very likely that there is at least one edge between $N_v$ and $V_{\ge w}$, and from there on it is likely to find a greedy path to the heavy core by Lemma~\ref{lem:path-to-core}. So it suffices to show that the second case is unlikely. Indeed, suppose that the second occurs for a large fraction of the vertices $v \in V_{< w}$. In this case we can carefully choose a set $V_u \subseteq V_{< w}$ that has in particular the property that the sets $N_v$ for $v \in V_u$ are disjoint. Then the vertices $u\in V_u$ have a significantly increased probability that $\deg(u)$ is large, which can only happen with very small probability by Lemma~\ref{lem:expecteddegree}. This is the most technical part of the proof, and the vague statements above are made precise by Claims~\ref{cla:claim2} and~\ref{cla:disjointNv} below. We can deduce that it is very unlikely that the second case happens for a large fraction of $V_{<w}$, from which the statement follows. \smallskip

We may assume $w \leq n^{1/2}$, since otherwise $k> n$, and the statement is trivial. Let $c>0$ be such that for all vertices $u$ of weight at least $w$, all vertices $u' \in V$, and every fixed position $\x {u'} \in \Space$ we have $\Pr[u \sim u' \mid \x{u'}] \geq cw/n$, i.e., $c$ is the hidden constant of condition~(EP1). Finally by the power-law assumption (PL2), for any sufficiently small $0<\eta<1$ we may choose $\tilde w = O(w^{1+\eta})$ such that there are at least $\Omega(n/w^{\beta-1+\eta})$ vertices with weights between $w$ and $\tilde w$.

We first uncover the graph $G_{< w}$ induced by vertices of weight less than $w$, i.e., we uncover the positions of these vertices and the edges in the induced subgraph. Let $v \in V_{<w}$ and let $N_v := N_v(k,w)$ be the $k$-neighborhood of $v$ in $G_{< w}$. Once $G_{< w}$ is fixed, consider a random vertex $u$, conditioned on $\w u \in [w,\tilde w]$. Let $R := R(v) := \Pr_u[u \sim N_v \mid G_{< w}]$. 

\begin{claim}\label{cla:Qlowerbound}
$Q := \Pr_u[|N_v \cap \Gamma(u)| \geq cw|N_v|/(2nR) \mid G_{< w} ] \geq \frac{cw}{2n}$.
\end{claim}
\begin{proof}
Let $x:= cw|N_v|/(2nR)$. We first use $|N_v \cap \Gamma(u)| \leq |N_v|$ to bound 
\[
\Ex\big[|N_v \cap \Gamma(u)| \big| G_{< w} \big] \leq Q|N_v| + (R-Q)x \leq Q|N_v|+Rx.
\]
On the other hand,  the left hand side is at least $cw|N_v|/n$ by our choice of $c$. Together, it follows $Q \geq cw/n -Rx/|N_v| = cw/(2n)$, proving the claim.
\end{proof}

Now we distinguish three cases for the vertex $v$. (1) If $|N_v| < k$ then obviously there is nothing to show. (2) If $R\geq w^{\beta}/n$, then 
$$\Ex\big[|\{ u \text{ with } \w{u}\in [w,\tilde w]\text{ and }u \sim N_v\}| \big| G_{< w}, R\geq w^{\beta}/n\big] \geq \Omega\Big(\frac{w^{\beta}}{n}\cdot \frac{n}{w^{\beta-1+\eta}}\Big) = \Omega(w^{\Omega(1)}),$$
because the number of vertices of weight in $[w,\tilde w]$ is at least $\Omega(n/w^{\beta-1+\eta})$ by (PL2). Since every $u$ draws its position and its edges to $V_{\le w}$ independently from each other, we may apply the Chernoff bounds and obtain 
\begin{equation}\label{eq:largeR}
\Pr\big[\exists u \text{ with } \w{u}\in [w,\tilde w]\text{ and }u \sim N_v \big| G_{< w}, R\geq w^{\beta}/n\big] \geq 1-O\big(e^{-w^{\Omega(1)}}\big),
\end{equation}
as desired.

(3) For the last case, $|N_v|\geq k$ and $R< w^{\beta}/n$, we will show that it is very unlikely that this case occurs for a random $v$ (over a random choice in $V_{<w}$). More precisely, let $V_R\subseteq V_{<w}$ be the set of vertices $v$ of weight less than $w$ for which $|N_v|\geq k$ and $R(v) < w^{\beta}/n$. Further, let $\mathcal{E}$ be the event that $|V_R| \geq ne^{-c'w}$, where $c'$ is a constant to be fixed later. Then we will show that $\Pr[\mathcal{E}] = e^{-\Omega(w)}$. Note that with this statement, we can conclude the proof as follows. Let $v$ be a random vertex of weight less than $w$. When we uncover $G_{<w}$, then $\mathcal{E}$ occurs only with probability $e^{-\Omega(w)}$. On the other hand, if $\mathcal{E}$ does not occur, then there at most $ne^{-c'w}$ vertices $v' \in V_{<w}$ for which $|N_{v'}|\geq k$ and $R(v')< w^{\beta}/n$, and the probability that $v$ is among them is at most 
\[\frac{ne^{-c'w}}{|V_{<w}|} = O\Big(\frac{ne^{-c'w}}{n(1-w^{1-\beta+\eta})}\Big)= O(e^{-\Omega(w)})\] for any $\eta>0$. Finally, if $v$ is not among these vertices, then either $|N_v| < k$, and we are done, or $R(v) \geq w^{\beta}/n$, and then $N_v \not\sim V_{\geq w}$ with probability at most $O(e^{-w^{\Omega(1)}})$ by~\eqref{eq:largeR}. Thus the theorem follows by a union bound, and it remains to show the following claim.

\begin{claim} \label{cla:claim2}
Let $V_R:= \{v \in V_{<w} \mid |N_v|\geq k \text{ and }R=R(v)< w^{\beta}/n\}$ and denote by $\mathcal{E}$ the event that $|V_R| \geq ne^{-c'w}$. Then
\begin{equation}
\Pr[\mathcal{E}] = O(e^{-\Omega(w)}).
\end{equation}
\end{claim}

Before we prove Claim~\ref{cla:claim2}, we need some preparation. Sort the vertices $v \in V_R$ decreasingly by $|N_v|$. We go through the list one by one, and pick greedily a set $V_{Gr} \subseteq V_R$ such that the $N_v$, $v\in V_{Gr}$ are pairwise disjoint. Then after this procedure, the following holds.
\begin{claim}
\label{cla:disjointNv}
$\sum_{v \in V_{Gr}} 2|N_{v}|^5 \geq |V_R|$.
\end{claim}
\begin{proof}[Proof of Claim~\ref{cla:disjointNv}]
We prove Claim~\ref{cla:disjointNv} by the following charging argument. Whenever we pick a vertex $v$ to be included into $V_{Gr}$, we inductively define levels $L_s(v) \subseteq V_R$, $s\geq 0$ by $L_0(v) := \{u \in N_{v} \mid |N_u| \leq |N_v|^2\}$ and 
\[L_{s+1}(v) := \bigcup_{v'\in L_s(v)}\big\{u \in N_{v'} \big\vert |N_u| \leq  |N_v|^{2^{-s}} \big\}.\] The vertex $v$ pays one coin to each vertex in $\cup_{s\geq 0} L_s(v)$. We claim that (i) every vertex $v$ that we pick pays at most $2|N_v|^5$ coins, and (ii) every vertex in $V_R$ is paid at least one coin. Note that (i) and (ii) together will imply Claim~\ref{cla:disjointNv}.

To prove (i), we observe that $|L_0(v)| \leq |N_v|$ and $|L_1(v)|/|L_0(v)| \leq |N_v|^2$ by definition of $L_0(v)$, and $|L_{s+1}(v)|/|L_s(v)| \leq |N_v|^{2^{1-s}}$ for all $s\geq 1$ by definition of $L_{s}(v)$. Therefore, $|L_s(v)| \leq |N_v|^{1+2+\sum_{j=1}^{s-1} 2^{1-j}}$. Moreover, for all $s > s_0 :=  \lfloor\log_2\log_k |N_v|\rfloor$ we have $|N_v|^{2^{-s}} < k$, so $L_{s+1} = \emptyset$ by definition of $V_R$. On the other hand, for all $s \leq s_0$ we have $|N_v|^{2^{-s}} \geq k \geq 2$, thus the terms $|N_v|^{3+\sum_{j=1}^{s-1} 2^{1-j}}$ increase at least geometrically fast for $s \leq s_0$. Hence, 
\[\sum_{s=0}^{s_0} |L_s(v)| \leq \sum_{s=0}^{s_0} |N_v|^{3+\sum_{j=1}^{s-1} 2^{1-j}} \leq 2 |N_v|^{3+\sum_{j=1}^{s_0-1} 2^{1-j}} \leq 2 |N_v|^5,\] proving (i).

For (ii), we show the following statement inductively for all vertices $v$. After $v$ has paid its coins, every vertex $u$ which comes after $v$ in the ordering, and for which $N_u \cap N_v \neq \emptyset$ holds, has received at least one coin. Note that it will follow that each vertex that we consider and that we do not pick has been paid by an earlier vertex. So assume that $u$ comes after $v$ in the ordering, and that $N_u \cap N_v \neq \emptyset$. Since we go through the vertices in descending order with respect to $|N_v|$, we have $|N_u|\leq |N_v|$. Let $v' \in N_u \cap N_v$. If $|N_{v'}| \leq |N_v|^2$, then $v'\in L_0$ and $u\in L_1$, so $v$ pays to $u$. If $|N_{v'}| > |N_v|^2$, then we have considered $v'$ before $v$. However, since we picked $v$, and since $v' \in N_v$ (and thus, $v\in N_{v'}$), $v'$ was not picked. Therefore, by induction hypothesis $v'$ had been paid by some earlier vertex $v''$, so $v'\in L_s(v'')$ for some $s \geq 0$. Since $|N_u| \leq |N_v| < |N_{v'}|^{1/2} \leq |N_{v''}|^{2^{-s}}$, we obtain $u \in L_{s+1}(v'')$, so $u$ has been paid by $v''$ as well. This proves (ii), and thus concludes the proof of Claim~\ref{cla:disjointNv}. Note that $2\sum_{v \in V_{Gr}} |N_v|^5 \geq |V_R| \geq ne^{-c'w}$ if $\mathcal{E}$ holds. 
\end{proof}

\begin{proof}[Proof of Claim~\ref{cla:claim2}]
With Claim~\ref{cla:disjointNv}, we can finally prove Claim~\ref{cla:claim2} as follows. Fix a vertex $u$ such that $\w{u}\leq \tilde w$. Then for each position $\x{u}$ of $u$, the expected degree of $u$ conditioned on $\x{u}$ is in $O(\tilde w)$, and it is the sum of independent random variables by Lemmas~\ref{lem:marginal} and~\ref{lem:expecteddegree}. Note that the hidden constant in the $O(\cdot)$-notation is independent of $\w{u}$ and of $\x{u}$. Therefore, by the Chernoff-Hoeffding bound (Theorem~\ref{thm:dubhashichernoff}), there are constants $c', C>0$ independent of $\w{u}$ and $\x{u}$ such that $\Pr[\deg(u)\geq i] \leq e^{-2c'i}$ for all $i \geq C\tilde w$, and this also holds if $u$ is a \emph{random} vertex with weight in $[w,\tilde w]$. So let $u$ be a random vertex with weight in $[w,\tilde w]$, and let 
\[V_u := \big\{v \in V_{Gr} \big\vert |N_v\cap\Gamma(u)|\geq |N_v|cw^{1-\beta}/2\big\}.\] 
Consider the random variables
\[
S_1 := 2\sum_{v \in V_u} |N_{v}|^5 \quad \text{and} \quad S_2 := \frac{cw^{1-\beta}}{2}\sum_{v \in V_u} |N_{v}|.\]
Note that $S_2 \leq \deg(u)$ by definition of $V_u$, and since all $v\in V_u\subseteq V_{Gr}$ have disjoint $N_v$ by construction. Hence, $\Pr[S_2 \geq i] \leq \Pr[\deg(u) \geq i] \leq e^{-2c'i}$ for all $i \geq C\tilde w$. Now consider the expectation of $S_1$ conditioned on $\mathcal{E}$. On the one hand, since we are in the case $R< w^{\beta}/n$, we have $|N_v|cw^{1-\beta}/2 < |N_v|cw/(2nR)$, and thus $\Pr[v \in V_u \mid v \in V_{Gr}] \geq cw/(2n)$ by Claim~\ref{cla:Qlowerbound}. Hence, by Claim~\ref{cla:disjointNv} we have 
\[\Ex[S_1 \mid \mathcal{E}] \geq \frac{cw}{n}\cdot \sum_{v \in V_{Gr}}|N_{v}|^5 \geq \frac{cw}{2}e^{-c'w}.\] 
On the other hand, since $\sum_{v\in V_u} |N_v|^5 \leq (\sum_{v\in V_u} |N_v|)^5$, we may lower-bound $$S_1 \leq 2 \cdot \left(2w^{\beta-1}S_2/c\right)^5.$$ Both inequalities together yield
\begin{align*}
\frac{cwe^{-c'w}}{2} & \leq \Ex[S_1 \mid \mathcal{E}] \leq 2\cdot \Big(\frac{2w^{\beta-1}}{c}\Big)^5\cdot \Ex[S_2^5 \mid \mathcal{E}] \\ & =  2\cdot \Big(\frac{2w^{\beta-1}}{c}\Big)^5\cdot \sum_{i\geq 1} i^5 \Pr[S_2=i \mid \mathcal{E}] \\
& \leq 2\cdot \Big(\frac{2w^{\beta-1}}{c}\Big)^5\cdot \frac{\sum_{i\geq 1} i^5 \Pr[S_2=i]}{\Pr[\mathcal{E}]}.
\end{align*}
Solving for $\Pr[\mathcal{E}]$ yields $\Pr[\mathcal{E}] \leq w^{O(1)}e^{c'w}\sum_{i\geq 1}i^5\Pr[S_2=i]$. Observe that $S_2 >0$ already implies $S_2 > cw^{1-\beta}k/2 =\omega(\tilde w)$, since $|N_{v}| \geq k$ for all $v \in V_{Gr}$. So if $w$ is sufficiently large then the first $C\tilde w$ terms of $\sum_{i\geq 1}i^5\Pr[S_2=i]$ vanish. On the other hand, recall that $\Pr[S_2 \geq i] \leq e^{-2c'i}$ for all $i \geq C\tilde w$. Hence, if $w$ is sufficiently large, 
$$\Pr[\mathcal{E}] \leq w^{O(1)}e^{c'w}\sum_{i\geq C \tilde w}i^5 \Pr[S_2 \geq i] \leq w^{O(1)}e^{c'w}\sum_{i\geq C \tilde w}i^5 e^{-2c'i}  = \tilde w^{O(1)}e^{-\Omega(\tilde w)} = O(e^{-\Omega(w)}).$$
This concludes the proof of Claim~\ref{cla:claim2}, and thus of the lemma.
\end{proof}
\end{proof}

The upper bounds on the diameter and the average distance now follow easily from the lemmas we proved so far. We collect the results in the following theorem, which reformulates and specifies Theorem~\ref{thm:diameter} and Theorem~\ref{thm:avgdist}.

\begin{theorem}[Components and Distances]\label{thm:components}\leavevmode
\begin{enumerate}[(i)]
\item Whp, there is a giant component, i.e., a connected component which contains $\Omega(n)$ vertices.
\item Whp, all other components have at most polylogarithmic size.
\item Whp, the giant component has polylogarithmic diameter.
\item In expectation and with probability $1-o(1)$, the average distance (i.e., the expected distance of two uniformly random vertices in the largest component) is $\frac{2 + o(1)}{|\log \beta-2 |}\log \log n$.
\item With probability $1-o(1)$, a $(1-o(1))$-fraction of all pairs of vertices in the giant component have distance at most $\frac{2 + o(1)}{|\log(\beta-2)|} \log \log n$.
\end{enumerate}
\end{theorem}
\begin{proof}(i) has been proven with Lemma~\ref{lem:path-to-core}~(iii). For (ii) and (iii) we fix a sufficiently small constant $\eps >0$ and conclude from the same lemma that whp the giant contains all vertices of weight at least $w:=(\log n)^C$, for a suitable constant $C>0$, and that whp all such vertices have distance at most $\frac{1+\eps}{|\log(\beta-2)|} \log \log n$ from the heavy core $\barV$. 
We apply Lemma~\ref{lem:path-to-core3} with  $\ell=w^{\beta+\eps}$. Then a random vertex in $V_{< w}$ has probability at least $1- e^{-w^{\Omega(1)}}$ to either be at distance at most $\ell$ of $V_{\geq w}$, or to be in a component of size less than $\ell$. Note that for sufficiently large $C$ this probability is at least $1-n^{-\omega(1)}$. By the union bound, whp one of the two options happens for all vertices in $V_{< w}$. This already shows that whp all non-giant components are of size less than $\ell=(\log n)^{O(1)}$. For the diameter of the giant, recall that whp the heavy core has diameter $o(\log \log n)$ by Lemma~\ref{lem:diameter-core}. Therefore, whp the diameter of the giant component is $O(\ell+\log \log n) = (\log n)^{O(1)}$.

For the average distance, let $\eps >0$, and let $v \in V$ be a vertex chosen uniformly at random. Fix $\ell\geq 3$, $\ell = n^{o(1)}$, and let $w := w(\ell) = \ell^{1/(\beta+1)}$. We sort the vertices by weight and uncover the graph vertex by vertex in increasing order, until either (1) we see for the first time a vertex $v' \in V_{\ge w}$ such that in the subgraph induced by $V_{\le \w {v'}}$ there exists a path of length at most $\ell$ from $v$ to $v'$, or (2) we have uncovered the full graph and (1) never happened. If $\w v \ge w$, then (1) trivially occurs. Otherwise, by Lemma~\ref{lem:path-to-core3}, with probability $1-O(\exp(-w^{\Omega(1)}))$ either case (1) happens or the connected component of $v$ in $G$ has size less than $\ell$. In the latter case, $v$ is not connected to the core and there is nothing to show. Otherwise, we have uncovered only the vertices of weight at most $\w {v'}$, which allows us to apply Lemma~\ref{lem:path-to-core}~(i) since its statement only depends on vertices of \emph{higher} weight. By Lemma~\ref{lem:path-to-core}~(i), with probability $1-O(\exp(-w^{\Omega(\eps)}))$ there is a weight-increasing path from $v'$ to the heavy core of length at most $\lambda_\eps := (1+\eps) \frac{\log \log n}{|\log(\beta-2)|}$. Summarizing, we have shown that for a random vertex $v$ and every $\ell\geq 3$ with $\ell = n^{o(1)}$
\begin{equation}\label{eq:coredistance}
\Pr\big[\infty > \dist(v,V_{\text{core}}) \geq \ell+ \lambda_\eps \big] \leq  e^{-\Omega(w(\ell)^{\Omega(\eps)})} = O\big(e^{-\ell^{\Omega(\eps)}}\big).
\end{equation}

Let us first consider the expectation of the average distance, i.e., if $u,u'$ denote random vertices in the largest component of a random graph $G$ then we consider $\Ex_G[ \Ex_{u,u'} [\dist(u,u') ] ]$. Since $\dist(u,u') \le n$ we can condition on any event happening with probability $1 - n^{-\omega(1)}$, in particular we can condition on the event $\mathcal E$ that $G$ has a giant component containing $V_{\text{core}}$, all other components have size $(\log n)^{O(1)}$, $G$ has diameter $(\log n)^{O(1)}$, and finally the core has diameter $d_{\text{core}} = o(\log \log n)$. Moreover, by bounding 
\[\dist(u,u') \le \dist(u,V_{\text{core}}) + \dist(u',V_{\text{core}}) + d_{\text{core}}\] 
it suffices to bound $2 \cdot \Ex_G[ \Ex_u[ \dist(u,V_{\text{core}})] \mid \mathcal E] + d_{\text{core}}$.
Now, since $\Ex[X] = \sum_{\ell> 0} \Pr[X\geq \ell]$ holds for a random variable $X$ taking values in $\mathbb{N}_{\ge 0}$, we can bound 
\[ \Ex_u[\dist(u,V_{\text{core}})] \le \lambda_\eps + \sum_{\ell = 1}^{(\log n)^{O(1)}} \Pr_u\big[ \dist(u,V_{\text{core}}) \ge \ell + \lambda_\eps \big]. \]
Note that conditioned on $\mathcal E$, since $u$ is chosen uniformly at random from the giant component, $\dist(u,V_{\text{core}}) < \infty$. Taking expectation over $G$, conditioned on $\mathcal E$, we may use \eqref{eq:coredistance} to bound the probability that $\dist(v,V_{\text{core}})$ is too large for a vertex chosen uniformly at random from $V$. Since the giant has size $\Omega(n)$, this probability increases at most by a constant factor if we instead choose $v$ uniformly at random from the giant. Hence, for every constant $\eps>0$ we obtain
\begin{equation}\label{eq:expavdist}
\Ex_G[ \Ex_u[\dist(u,V_{\text{core}})] \mid \mathcal E] 
  \le \lambda_\eps + \sum_{\ell = 1}^{(\log n)^{O(1)}} O(e^{-\ell^{\Omega(\eps)}}) + n^{-\omega(1)}. 
\end{equation}
We now use the inequality
\[ \sum_{\ell=1}^\infty e^{-\ell^\kappa} \le \int_{x=0}^\infty e^{-x^\kappa} d x = \Gamma(1+1/\kappa), \]
where $\Gamma$ is Euler's Gamma function. Since $\Gamma(x)$ is monotonicly increasing on the real axis for $x \ge 2$ and $\Gamma(1+n) = n!$, we have $\Gamma(1+1/\kappa) \le \lceil 1/\kappa \rceil ! \le (1/\kappa)^{O(1/\kappa)}$ for $\kappa \le 1$. Plugging this into equation~\eqref{eq:expavdist} yields
\[\Ex_G\big[ \Ex_u[\dist(u,V_{\text{core}})] \mid \mathcal E\big] 
  \le \lambda_\eps + O(1/\eps)^{O(1/\eps)} + n^{-\omega(1)}. \]
Note that for sufficiently slowly falling $\eps = \eps(n) = o(1)$ we have 
$O(1/\eps)^{O(1/\eps)} = o(\log \log n)$.
This yields the desired bound on the expected average distance of 
\[2\lambda_{o(1)} + o(\log \log n) = \frac{2+o(1)}{\log|\beta-2|}\log \log n.\]

For the concentration, we want to show $\Pr_G[ \Ex_{u,u'}[ \dist(u,u') ] \ge 2\lambda_{\eps}]=o(1)$, where we choose the same $\eps(n)=o(1)$ as before. We may upper-bound $\Pr_G[ \Ex_{u,u'}[ \dist(u,u') ] \ge 2\lambda_{\eps}]$, similarly as before, by
\[ n^{-\omega(1)} + \Pr_G\big[ 2\cdot \Ex_u[\dist(u,V_{\text{core}})] + d_{\text{core}} \ge 2\lambda_{\eps} \mid \mathcal E \big].\]
Let $\gamma > 0$ be a sufficiently small constant, and let 
\[\rho  =\rho(n) = \frac{\eps}{3|\log(\beta-2)|} \cdot \log \log n = o(\log \log n).\] 
We claim that for sufficiently large~$n$, $2\cdot \Ex_u[\dist(u,V_{\text{core}})] + d_{\text{core}} \ge 2\lambda_{\eps}$ can only happen if for some $\ell > \rho$ we have 
\begin{equation}\label{eq:avgdistsomecondition}
\Pr_u[\dist(u,V_{\text{core}}) \ge \ell + \lambda_{\eps/3}] \ge e^{-2\ell^{\gamma \cdot \eps}}.
\end{equation}
Indeed, otherwise we have (conditioned on $\mathcal E$), similarly as before
$$
\Ex_u[\dist(u,V_{\text{core}})] \le \lambda_{\eps/3} + \rho + \sum_{\ell = \rho}^{(\log n)^{O(1)}} \Pr_u\big[\dist(u,V_{\text{core}}) \ge \ell + \lambda_{\eps/3}\big] \le \lambda_{2\eps/3} + O(1/\eps)^{O(1/\eps)}, 
$$
and thus, indeed $\Ex_{u,u'}[\dist(u,u')]$ is at most
\[2 \cdot \Ex_u[\dist(u,V_{\text{core}})] + d_{\text{core}} \le 2 \lambda_{2\eps/3} + o(\log \log n) < 2 \lambda_{\eps}=2\lambda_{o(1)},\] 
if $\eps=\eps(n) = o(1)$ decreases sufficiently slowly.
However, using the union bound over all $\rho \le \ell \le (\log n)^{O(1)}$, the probability that $G$ is such that \eqref{eq:avgdistsomecondition} holds for some $\ell > \rho$ is bounded from above by
\[ \sum_{\ell = \rho}^{(\log n)^{O(1)}} \Pr_G\big[ \Pr_u[\dist(u,V_{\text{core}}) \ge \ell + \lambda_{\eps/3}] \ge e^{-\ell^{\gamma \cdot \eps}} \big| \mathcal E \big]. \]
By \eqref{eq:coredistance} it follows that 
\[\Ex_G[\Pr_u[\infty > \dist(u,V_{\text{core}}) \ge \ell + \lambda_{\eps/3}]] \le O(\exp(-2\ell^{\gamma\cdot\eps})),\] for $\gamma>0$ sufficiently small. We apply Markov's inequality and deduce that 
\[\Pr_G\big[\Pr_u[\infty > \dist(u,V_{\text{core}}) \ge \ell + \lambda_{\eps/3}] > e^{-\ell^{\gamma \cdot \eps}} \big] \le O(e^{-\ell^{\Omega(\eps)}}).\]

Because the giant has linear size, this probability increases at most by a constant factor if we instead draw $v$ from the giant component (conditioned on $\mathcal E$). Thus, the desired probability is bounded by
\begin{equation}\label{eq:avgdistlastbound}
\sum_{\ell = \rho}^{(\log n)^{O(1)}} O(e^{-\ell^{\Omega(\eps)}}), 
\end{equation}
which is $o(1)$, since $\rho = \eps \cdot \log \log n$ grows sufficiently quickly compared to a sufficiently slowly falling $\eps = o(1)$. This shows the concentration of the average distance and proves statement~(iv). 

Regarding the last statement~(v), \eqref{eq:avgdistlastbound} shows that with probability $1-o(1)$ $G$ is such that \eqref{eq:avgdistsomecondition} does not hold for any $\ell > \rho$. However, in this case the fraction of pairs $\{u,u'\}$ of vertices in the giant that have distance at least $2\ell + 2\lambda_{\eps/3}$ is at most $e^{-4\ell^{\gamma\cdot\eps}}$. By taking $\ell=2\rho$ and assuming that $\rho=\eps\log\log n$ grows sufficiently quickly compared to $\eps=o(1)$, we see that a $(1-o(1))$-fraction of pairs $\{u,u'\}$ has distance at most $2\lambda_{o(1)}$, given that $\rho=\eps \cdot \log\log n$ grows sufficiently fast compared to $\eps=o(1)$. This finishes the proof of Theorem~\ref{thm:components}.
\end{proof}

\section{Degree Sequence}
\label{sec:degreesequence}

By definition of the model, we are assuming that the weight sequence $\w{}$ follows a power law. Since the expected degree of a vertex with weight $\w v$ is $\Theta(\w v)$ by Lemma~\ref{lem:expecteddegree}, it is not surprising that the degree sequence of the random graph will also follow a power law. In this section, we give details and prove Theorem~\ref{thm:degseq1}, where we use Theorem~\ref{thm:concentration} for showing concentration. Some ideas of our proof are based on \cite{gugelmann2012random}. We start with the maximum degree $\Delta(G)$, which is a simple corollary of Lemma~\ref{lem:largevertices}.

\begin{corollary} \label{cor:maxdegree}
Whp, 
$\Delta(G)=\Theta(\wmax)$, where $\wmax = \max\{\w v \mid v \in V\}$. In particular,  for all $\eta>0$, whp, $\Delta(G)=\Omega(\barw)$ and $\Delta(G)=O(n^{1/(\beta-1-\eta)})$.
\end{corollary}

\begin{proof}
We deduce from the model definition that $\omega(\log^2 n) \le \barw \le \wmax =O(n^{1/(\beta-1-\eta)})$. Then Lemma~\ref{lem:largevertices} directly implies the statement.
\end{proof}

Next, we calculate the expected number of vertices that have degree at least $d$.

\begin{lemma} \label{lem:expnrverticesdegree}
Let $\eta>0$ be sufficiently small. Then for all $d \ge 1$, $d=d(n)=o(\barw)$, we have
$$\Omega\big(n d^{1-\beta-\eta}\big) \le \Ex[\#\{v \in V \mid \deg(v) \geq d\}] \le O \big(n d^{1-\beta+\eta}\big).$$
\end{lemma}
\begin{proof}
Let $\eta>0$ be sufficiently small. Recall that by Lemma~\ref{lem:expecteddegree}, it holds $\Ex[\deg(v)]=\Theta(\w v)$ for every vertex $v \in V$. Let $1 \le d \ll \barw$ and let $v$ be any vertex with weight $w_v \ge \Omega(d)$ large enough such that $\Ex[\deg(v)] \ge 2d$. Then by a Chernoff bound
$$\Pr[\deg(v) < d] \le \Pr[\deg(v) < 0.5 \Ex[\deg(v)]] \le e^{-\Ex[\deg(v)]/8} \le e^{-d/4} \le e^{-1/4}.$$ 
By the power-law assumption (PL2) there are $\Omega(n d^{1-\beta-\eta})$ vertices with weight $\Omega(d)$, and a single vertex of this set has degree at least $d$ with probability at least $1-e^{-1/4}$. By linearity of expectation, $\Ex[\#\{v \in V \mid \deg(v) \geq d\}]=\sum_{v \in [n]} \Pr[\deg(v)\geq d]=\Omega(n d^{1-\beta-\eta})$.

Next let $v$ be a vertex with weight $\w v \le O(d)$ small enough such that $2e\Ex[\deg(v)] \le 3d/4$. By a Chernoff bound (Theorem~\ref{thm:dubhashichernoff}.(iii)) we obtain 
$$\Pr[\deg(v) \ge d] \le \Pr[\deg(v) > 3d/4] \le 2^{-3d/4}.$$  
Thus, for the upper bound it follows
\begin{align*}
\Ex[\#\{v \in V \mid \deg(v) \geq d\}]&=\sum_{v \in [n]} \Pr[\deg(v)\geq d] \\&\le \vert V_{\ge O(d)}\vert + \!\!\!\!\sum_{v \in V_{\le O(d)}}\!\! \Pr[\deg(v)\geq d]\\
&\le O(n d^{1-\beta+\eta}) + n\cdot 2^{-3d/4}.
\end{align*}
Note that $d^2 \le 3 \cdot 2^{3d/4}$ holds for all $d \ge 1$. Hence
$n\cdot 2^{-3d/4} \le 3nd^{-2} < 3nd^{1-\beta+\eta}$ and indeed  it holds $\Ex[\#\{v \in V \mid \deg(v) \geq d\}] =O( n d^{1-\beta+\eta})$.
\end{proof}

After these preparations we come to the main theorem of this section which is a more precise formulation of Theorem~\ref{thm:degseq1} and states that the degree sequence follows a power law with the same exponent $\beta$ as the weight sequence.

\begin{theorem} \label{thm:degsequence}
For all $\eta>0$, whp we have
$$\Omega\big(n d^{1-\beta-\eta}\big) \le \#\{v \in V \mid \deg(v) \geq d\} \le O \big(n d^{1-\beta+\eta}\big),$$
where the first inequality holds for all $1 \leq d \leq \barw$ and the second inequality holds for all $d \ge 1$.
\end{theorem}

Before we prove Theorem~\ref{thm:degsequence}, we note that together with the standard calculations from Lemma~\ref{lem:totalweight} we immediately obtain the average degree in the graph.
\begin{corollary} \label{cor:averagedegree}
With high probability, $\frac{1}{n}\sum_{v \in V} \deg(v)= \Theta(1)$ and thus $|E|=\Theta(n)$.
\end{corollary}

\begin{proof}[Proof of Theorem~\ref{thm:degsequence}]

We first consider the case where $d$ is larger than $\log^3 n = o(\barw)$. From Condition (PL2) on the vertex weights and Lemma~\ref{lem:expecteddegree} it follows that 
$$\#\{v \in V \mid \Ex[\deg(v)] \geq 1.5d\}=\Omega\big(nd^{1-\beta-\eta}\big)$$ holds for all $\log^3 n \le d \le \barw$. Then by Lemma~\ref{lem:largevertices}, whp every vertex $v$ with $\Ex[\deg(v)] \geq 1.5d$ has degree at least $(1-o(1))1.5d \ge d$ for $n$ large enough. Hence whp there exist at least $\Omega(nd^{1-\beta-\eta})$ vertices with degree at least $d$. 
Vice-versa, by Lemma~\ref{lem:expecteddegree} we have
$$\#\{v \in V \mid \Ex[\deg(v)] \geq 0.5d\} =O\big(n d^{1-\beta+\eta}\big).$$ 
By the same arguments as above, whp every vertex $v$ with $\Ex[\deg(v)] < 0.5d$ has degree at most $(1+o(1))0.5d < d$. Thus in total there can be at most $O(n d^{1-\beta+\eta})$ vertices with degree at least $d$. This proves the theorem for $d \ge \log^3 n$.

Let $1 \le d \le \log^3 n$, $\eps>0$ be sufficiently small, $V' := V_{\le n^{\eps}}$ be the set of small-weight vertices, and $G' := G[V']$. First, we introduce some notation and define the two random variables 
$$g_d := \#\{v\in V \mid \deg(v) \geq d\}\quad \text{and}\quad f_d := \#\{v\in V' \mid \deg_{G'}(v) \geq d\}.$$ Note that by Lemma~\ref{lem:expnrverticesdegree}, we already have 
\[\Omega\big(n d^{1-\beta-\eta}\big) \le \Ex[g_d] \le O\big(n d^{1-\beta+\eta}\big),\]
and it remains to prove concentration. Clearly,
\begin{equation} \label{eq:helpingineq}
f_d \le g_d \le f_d + 2\sum_{v \in V\setminus V'} \deg(v).
\end{equation}
Next we apply Lemma~\ref{lem:largevertices} together with Lemma~\ref{lem:totalweight} and see that whp,
$$\sum_{v \in V\setminus V'} \deg(v) = \Theta\big(\W_{\ge n^{\eps}} \big) = O\big(n^{1+(2-\beta+\eta)\eps}\big)=n^{1-\Omega(1)}.$$
Recall that we assume $d \le \log^3 n$, so in particular $\Ex[g_d] =\Omega( n / (\log n)^{3(\beta-1+\eta)})$.
It follows that $\Ex\big[\sum_{v \in V\setminus V'} \deg(v)\big] = o(\Ex[g_d])$. Inequality (\ref{eq:helpingineq}) thus implies $\Ex[f_d]=(1+o(1))\Ex[g_d]$. Hence, it is sufficient to prove that the random variable $f_d$ is concentrated around its expectation, because this will transfer immediately to $g_d$.

We aim to show this concentration result via Theorem~\ref{thm:concentration}. Similarly to the proof of Claim~\ref{cla:nrbadvertices}, we can assume that the considered probability space $\Omega$ is a product space of independent random variables. More precisely, the $n$ independent random variables $\x 1, \ldots, \x n$ define the vertex set and the $n-1$ independent random variables $Y_2, \ldots, Y_n$ define the edge set, where each $Y_u$ has the form $(Y_u^1, \ldots, Y_u^{u-1})$, each $Y_u^v$ is a real number chosen uniformly at random from $[0,1]$, and for $v < u$, the edge $\{u,v\}$ is present in the graph if and only if $p_{uv}>Y_u^v$. The $2n-1$ random variables then define the product probability space $\Omega$, i.e., for every $\omega \in \Omega$, we denote by $G(\omega)$ the resulting graph, and similarly we use $G'=G'(\omega)$  and $f_d = f_d(\omega)$. We now consider the bad event:
\begin{equation} \label{eq:badevent}
\mathcal{B} := \{\omega \in \Omega: \text{ the maximum degree in  }G'(\omega) \text{ is at least }n^{2\eps}\}.
\end{equation}
We observe that $\Pr[\mathcal{B}]=n^{-\omega(1)}$, since by Lemma~\ref{lem:largevertices} whp every vertex $v \in V'$ has degree at most $O(\w{v} + \log^2 n) = o(n^{2\eps})$. Let $\omega,\omega' \in \overline{\mathcal{B}}$ such that they differ in at most two coordinates. We observe that changing one coordinate $\x{i}$ or $Y_i$ can influence only the degrees of $i$ itself and of the vertices which are neighbors of $i$ either before or after the coordinate change. It follows that $|f_d(\omega)-f_d(\omega')| \le 4 n^{2\eps} =: c$. Therefore, $f_d$ satisfies the Lipschitz condition of Theorem~\ref{thm:concentration} with bad event $\mathcal{B}$. Let $t = n^{1-\eps}=o(\Ex[f_d])$. Then since $n \Pr[\mathcal B] = n^{-\omega(1)}$, Theorem~\ref{thm:concentration} implies
$$
\Pr\big[|f_d-\Ex[f_d]| \ge t \big] \le 2e^{-\frac{t^2}{64c^2 n}} + (\tfrac{4n^2}{c}+1)\Pr[\mathcal{B}] = e^{-\Omega(n^{1-4\eps})}+n^{-\omega(1)}=n^{-\omega(1)},
$$
which proves the concentration and concludes the proof. 
\end{proof}

\section{Example: GIRGs and generalizations}\label{sec:distancemodel}

In this section, we further discuss the special cases of our model mentioned in Section~\ref{sec:model}. Mainly, we study a class which is still fairly general, the so-called \emph{distance model}. We show that the GIRG model introduced in~\cite{bringmann2015euclideanGIRG} is a special case, and we also discuss a non-metric example. In addition, with the \emph{threshold model} we consider a variation which includes in particular threshold hyperbolic random graphs.

\paragraph{The distance model:} 

We consider the following situation, which will cover both GIRGs and the non-metric example. As our underlying
geometry we specify the ground space $\Space=[0,1]^d$, where $d \ge 1$ is a
(constant) parameter of the model. We sample from this set according to the
standard (Lebesgue) measure. This is in the spirit of the classical random
geometric graphs~\cite{penrose2003}. 

To describe the distance of two points $x,y \in \Space$, assume we have some measurable function $\|.\|: [-1/2,1/2)^d \to \R_{\geq 0}$ such that $\|0\| = 0$ and $\|-x\|=\|x\|$ for all $x\in [-1/2,1/2)^d$. Note that $\|.\|$ does not need to be a norm or seminorm. We extend $\|.\|$ to $\R^d$ via $\|z\| := \| z-u\|$, where $u\in \Z^d$ is the unique lattice point such that $z-u \in [-1/2,1/2)^d$. For $r\geq 0$ and $x \in \Space$, we define the $r$-ball around $x$ to be $B_r(x):=\{x \in \Space \mid \|x-y\| \leq r\}$, and we denote by $V(r)$ the volume of the $r$-ball around $0$. Intuitively, $B_r(x)$ is the ball around $x$ in $[0,1]^d$ with the torus geometry, i.e., with $0$ and $1$ identified in each coordinate. Assume that $V\colon \R_{\geq 0} \to [0,1]$ is surjective, i.e., for each $V_0 \in [0,1]$ there exists $r$ such that $V(r)=V_0$.

Finally let $\alpha \in \R_{>0}$ be a long-range parameter. Since the case $\alpha=1$ deviates slightly from the
general case, we assume $\alpha \neq 1$. 
Let $p$ be any edge probability function that satisfies for all $u,v$ and $\x u,\x v \in \Space = [0,1]^d$,
\begin{equation}\label{eq:distancemodel}
 p_{uv}(\x{u},\x{v}) = \Theta\Big( \min\Big\{1, V(\|\x u-\x v\|)^{-\alpha} \cdot \Big( \frac{\w u \w v}{\W}\Big)^{\max\{\alpha,1\}} \Big\} \Big).
\end{equation}
Then, as we will prove later in Theorem~\ref{thm:distancemodel}, $p$ satisfies conditions (EP1) and (EP2), so it is a special case of our model.


\begin{example} \label{ex:girg}
If we choose $\|.\|$ to be the Euclidean distance $\|.\|_2$ (or any equivalent norm such as $\|.\|_{\infty}$) then we obtain the \emph{GIRG} model introduced in~\cite{bringmann2015euclideanGIRG} and \cite{SerranoKB08}, where the distance of two points $x,y$ in $[0,1]^d$ is given by their geometric distance on the torus. 
In~\cite{bringmann2015euclideanGIRG} it was shown that a graph from such a GIRG model whp has clustering coefficient $\Omega(1)$, that it can be stored with $O(n)$ bits in expectation, and that it can be sampled in expected time $O(n)$. Moreover, it was shown that hyperbolic random graphs are contained in the $1$-dimensional GIRG model. Recently, processes such as bootstrap percolation \cite{koch2016bootstrap} and greedy routing \cite{bringmann2017routing} have been analyzed on this model.
\end{example}

The next distance measure is particularly useful to model social networks: assume that two individuals share one feature (e.g., they are in the same sports club), but are very different in many other features (work, music, ...). Then they are still likely to know each other, which is captured by the minimum component distance.

\begin{example}
Let the \emph{minimum component distance} be defined by 
$$\|x\|_{\min} := \min\{x_i \mid 1\leq i \leq d\} \text{ for } x = (x_1,\ldots,x_d) \in
[-1/2,1/2)^d.$$ 
Note that the minimum component distance is not a metric for $d\geq 2$, since there are
$x,y,z\in \Space$ such that $x$ and $y$ are close in one component, $y$ and $z$
are close in one (different) component, but $x$ and $z$ are not close in any
component. Thus the triangle inequality is not satisfied. However, it still satisfies the requirements specified above, so our results of this paper apply.
\end{example}

\begin{theorem}\label{thm:distancemodel}
In the geometric setting described above, let $p$ be any function that satisfies Equation~(\ref{eq:distancemodel}). Then conditions (EP1) and (EP2) are satisfied, and we obtain an instance of the general model.
\end{theorem}
 \begin{proof}
Fix $u,v$, and also the position $\x u$. Note that $V(r)$ is the cumulative probability distribution $\Pr_{\x v}(\|\x u-\x v\|\leq r)$. The marginal edge probability is given by the Riemann-Stieltjes integral over $r$,
\[E := \Ex_{\x v}[p_{uv}(\x u,\x v) \mid \x u] =  \Theta\Big( \int_0^\infty \Lambda_{u,v}(r) dV(r) \Big),\]
where
\[\Lambda_{u,v}(r) := \min\Big\{1, V(r)^{-\alpha} \cdot \Big( \frac{\w u \w v}{\W}\Big)^{\max\{\alpha,1\}} \Big\}.\]
In particular, for every sequence of partitions $r^{(t)} = \{0=r_0^{(t)} < \ldots <r_{\ell(t)}^{(t)}\}$ with meshes tending to zero, the upper Darboux sum with respect to $r^{(t)}$ converges to the expectation,
\[E =  \Theta\Big(\lim_{t\to \infty} \sum_{s=1}^{\ell(t)}\Big(\sup_{r_{s-1}^{(t)}\leq  r\leq r_{s}^{(t)}} \Lambda_{u,v}(r)\Big) \big(V(r_{s}^{(t)})-V(r_{s-1}^{(t)})\big) \Big).\]
Since $V$ is surjective, we may refine the meshes $r^{(t)}$ if necessary such that the meshes of the partitions $V^{(t)} = \{V(r_0^{(t)}), \ldots, V(r_{\ell(t)}^{(t)})\}=:\{V_0^{(t)}, \ldots, V_{\ell(t)}^{(t)}\}$ also tend to zero. Hence,
\begin{align*}
E & = \Theta\Big(\lim_{t\to \infty} \sum_{s=1}^{\ell(t)}\min\Big\{1, (V_s^{(t)})^{-\alpha} \cdot \Big( \frac{\w u \w v}{\W}\Big)^{\max\{\alpha,1\}} \Big\} \Big(V_{s}^{(t)}-V_{s-1}^{(t)}\Big)\Big) \\
& = \Theta\Big(\int_{0}^1  \min\Big\{1, V^{-\alpha} \cdot \Big( \frac{\w u \w v}{\W}\Big)^{\max\{\alpha,1\}} \Big\} dV\Big),
\end{align*}
where the latter integral is an ordinary Riemann integral. If $\w u \w v/\W \geq 1$, the integrand is $1$ and we obtain $E = \Theta(1) = \Theta\left(\min\left\{1,\frac{\w u \w v}{\W}\right\}\right)$. On the other hand, if $\w u \w v/\W < 1$ then let $r_0:= (\frac{\w u \w v}{\W})^{\max\{\alpha,1\}/\alpha } < 1$. Note that if $r_0 = \Theta(1)$, then also $r_0 = \Theta(\w u \w v/\W)$. Therefore,
\begin{align*}
E & = \Theta\bigg(\int_0^{r_0} 1 dV+\Big( \frac{\w u \w v}{\W}\Big)^{\max\{\alpha,1\}}\int_{r_0}^1 V^{-\alpha} dV\bigg) \\
& = \begin{cases} 
\Theta\Big(r_0 + \frac{\w u \w v}{\W}(1-r_0^{1-\alpha})\Big) = \Theta \Big( \frac{\w u \w v}{\W}\Big)& \text{, if } \alpha < 1, \text{ and}\\
\Theta\bigg(r_0 +\Big( \frac{\w u \w v}{\W}\Big)^{\alpha}(r_0^{1-\alpha}-1) \bigg)  = \Theta \Big( \frac{\w u \w v}{\W}\Big)& \text{, if } \alpha> 1, 
\end{cases} 
\end{align*}
as required.

It remains to show that $p$ satisfies~(EP2). Since $V(\|\x u-\x v\|) \leq 1$, from Equation~\eqref{eq:distancemodel} we obtain the lower bound 
\[p_{uv} \geq \Omega\Big( \min\Big\{1, \Big(\frac{\w u \w v}{\W}\Big)^{\max\{\alpha,1\}}\Big\}\Big).\]
If $\w u \w v/\W \geq 1$ then there is nothing to show (since the right hand side of (EP2) is $o(1)$ by the upper bound on $\barw$). Otherwise, if $\w u \w v/\W < 1$, then
\[
p_{uv} \geq \Omega\Big( \Big(\frac{\w u \w v}{\W}\Big)^{\max\{\alpha,1\}}\Big) \geq \Omega\Big( \frac{\barw^2}{n}\Big) \geq \Big(\frac{n}{\barw^{\beta-1+\eta}}\Big) ^{-1+\omega(1/\log\log n)},
\]
where the last step follows from the lower bound on $\barw$. This concludes the proof.
\end{proof}

\paragraph{The threshold model:} 

Finally, we discuss a variation of Example~\ref{ex:girg} where we let $\alpha \to \infty$ and thus obtain a threshold function for $p$.

\begin{example} \label{ex:threshold}
Let $\|.\|$ be the Euclidean distance $\|.\|_2$ and let $p$ again satisfy (\ref{eq:distancemodel}), but this time we assume that $\alpha=\infty$. More precisely, we require
\begin{equation}\label{eq:puv2}
 p_{uv}(\x u,\x v) = \begin{cases} \Theta(1) & \text{if } \|\x u - \x v\| \le O\big(\big(\tfrac{\w u \w v}\W\big)^{1/d}\big) \\ 0 & \text{if } \|\x u - \x v\| \ge \Omega\big(\big(\tfrac{\w u \w v}\W\big)^{1/d}\big), \end{cases}  
 \end{equation}
where the constants hidden by $O$ and $\Omega$ do not have to match, i.e., there can be an interval $[c_1 (\tfrac{\w u \w v}\W)^{1/d}, c_2 (\tfrac{\w u \w v}\W)^{1/d}]$ for $\|\x u - \x v\|$ where the behaviour of $p_{uv}(\x u,\x v)$ is arbitrary. This function $p$ yields the case $\alpha = \infty$ of the \emph{GIRG} model introduced in~\cite{bringmann2015euclideanGIRG}. In~\cite{bringmann2015euclideanGIRG} we proved that threshold hyperbolic random graphs are contained in this model, and furthermore that the model whp has clustering coefficient $\Omega(1)$, it can be stored with $O(n)$ bits in expectation, and that it can be sampled in expected time $O(n)$.

Notice that the volume of a ball with radius $r_0=\Theta((\frac{\w u \w v}{\W})^\frac{1}{d})$ around any fixed $x \in \Space$ is $\Theta(\min\{1,\frac{\w u \w v}{\W}\})$. Thus, by (\ref{eq:puv2}), for fixed $\x{u}$ it follows directly that
$$\Ex_{\x v}[p_{uv}(\x u,\x v) \mid \x u] = \Theta\big(\Pr_{\x v}\big[\|\x u - \x v\| \le r \mid \x u\big] \big) = \Theta\left(\min\left\{1,\tfrac{\w u \w v}{\W}\right\}\right).$$
Since (EP1) is satisfied, Theorem~\ref{thm:degseq1} for the degree sequence already applies. In order to also fulfill (EP2), we additionally require that $2 < \beta < 3$ and $\barw=\omega(n^{1/2})$. Then for all $\w u, \w v \ge \barw$ we have $\frac{\w u \w v}{\W} = \omega(1)$. For all positions $\x u,\x v \in \Space$ we thus obtain $p_{uv}(\x u,\x v)=\Theta(1)$ by (\ref{eq:puv2}). 
\end{example}

\begin{remark}\label{rem:largebeta}
It follows from the definition that the low-weight vertices in a GIRG contain ordinary random geometric graphs as subgraphs, i.e., every pair of vertices connects with probability $\Omega(1)$ if the distance between the vertices is at most $cn^{-1/d}$, where $c$ is a constant that depends on the minimal weight $\wmin$. If $\wmin$ is sufficiently large, then these subgraphs are supercritical, i.e., they have a giant component. On the other hand, in the threshold model for $\beta >3$ sufficiently large, all edges cover a polynomially small distance $n^{-\Omega(1)}$. Thus, by combining these conditions we get a random graph model in the regime $\beta >3$ with giant components where the average distance is polynomially large. 
\end{remark}

\section{Conclusion}
\label{sec:conclusion}

We studied a class of random graphs that genericly augment Chung-Lu random graphs by an underlying ground space, i.e., every vertex has a random position in the ground space and edge probabilities may arbitrarily depend on the vertex positions, as long as marginal edge probabilities are preserved. Since our model is very general, it contains well-known special cases like hyperbolic random graphs~\cite{BogunaPK10,PapadopoulosKBV10} and geometric inhomogeneous random graphs~\cite{bringmann2015euclideanGIRG}. Beyond these well-studied models, our model also includes non-metric ground spaces, which are motivated by social networks, where two persons are likely to know each other if they share a hobby, regardless of their other hobbies.

Despite its generality, we show that all instantiations of our model have similar connectivity properties, assuming that vertex weights follow a power law with exponent $2 < \beta < 3$. In particular, there exists a unique giant component of linear size and the diameter is polylogarithmic. Surprisingly, for all instantiations of our model the average distance is the same as in Chung-Lu random graphs, namely $(2 \pm o(1))\frac{\log \log n}{|\log(\beta-2)|}$. In some sense, this shows universality of ultra-small worlds. 

We leave it as an open problem to determine whether the diameter of our model is $O(\log n)$ for $2 < \beta < 3$.

\paragraph{Acknowledgements:}
We thank Hafsteinn Einarsson, Tobias Friedrich, and Anton Krohmer for helpful discussions.
\begin{footnotesize}
\bibliographystyle{plain}
\bibliography{../girg}
\end{footnotesize}

\end{document}